\documentclass[journal]{IEEEtran}

\usepackage[svgnames]{xcolor} 
\usepackage{pstricks,pst-node,pst-plot,pstricks-add}
\usepackage[tight,footnotesize]{subfigure}
\usepackage{tikz}

\usepackage{pgfplots}
\usepgfplotslibrary{groupplots,dateplot}
\usetikzlibrary{patterns,shapes.arrows,plotmarks,positioning,spy}

\usepackage{epsfig}
\usepackage{graphicx}
\usepackage{mathtools}
\usepackage{amsmath, amsthm}
\usepackage{amsfonts}
\usepackage{amssymb}
\usepackage{booktabs}
\usepackage{cite}
\usepackage{dsfont}
\usepackage{graphics}
\usepackage{xcolor}
\usepackage{enumerate}
\usepackage{clipboard}
\newclipboard{myclipboard}

\newcommand{\xv}{\boldsymbol{x}}
\newcommand{\yv}{\boldsymbol{y}}
\newcommand{\pv}{\boldsymbol{\phi}}
\newcommand{\lv}{\boldsymbol{\lambda}}
\newcommand{\zerov}{\boldsymbol{0}}
\newcommand{\thetav}{\boldsymbol{\theta}}
\newcommand{\cv}{\boldsymbol{c}}

\newcommand{\Cc}{\mathcal{C}}
\newcommand{\Xc}{\mathcal{X}}
\newcommand{\Yc}{\mathcal{Y}}
\newcommand{\RR}{{\mathbb{R}}}

\newtheorem{theorem}{Theorem}
\newtheorem{lemma}{Lemma}
\newtheorem{corollary}{Corollary}[theorem]
\newtheorem{definition}{Definition}

\newcommand{\comment}[1]{}
\newcommand{\QX}{Q_{\!X\!}}
\newcommand{\EX}{{\mathbb{E}}}
\newcommand{\VAR}{{\text{\normalfont Var}}}
\newcommand{\DKL}{{D}}
\newcommand{\DD}{{\mathrm{d}}}

\DeclareMathOperator*{\argmax}{argmax}

\begin{document}
\allowdisplaybreaks

\title{Achievable Rates and Error Exponents\\ for a Class of Mismatched Compound Channels}

%\author{Priyanka~Patel,	Francesc~Molina and~Albert~Guill\'en i F\`abregas
\author{Priyanka~Patel, Francesc~Molina and~Albert~Guill\'en i F\`abregas
\thanks{Priyanka Patel was with the Department of Engineering, University of Cambridge, CB2 1PZ Cambridge, UK, and is now with Google, London; e-mail: {\tt pp490@cantab.ac.uk}.

Francesc Molina was with the Department of Engineering, University of Cambridge, CB2 1PZ Cambridge, UK, and with the Department of Signal Theory and Communications, Universitat Polit\`ecnica de Catalunya 08034 Barcelona, Spain, and is now with the Department of Information and Communication Technologies, Universitat Pompeu Fabra, 08018 Barcelona, Spain; e-mail: {\tt fm585@cam.ac.uk}.

Albert Guill\'en i F\`abregas is with the Department of Engineering, University of Cambridge, CB2 1PZ Cambridge, UK, and with the Department of Signal Theory and Communications, Universitat Polit\`ecnica de Catalunya, 08034 Barcelona, Spain; e-mail: {\tt guillen@ieee.org}.
}
\thanks{This work was presented at the 2024 International Z\"urich Seminar on Information and Communication, Z\"urich, Switzerland, and at the 2024 IEEE International Symposium on Information Theory, Athens, Greece. This work was supported in part by the European Research Council under Grants 725411 and 101142747, and in part by the Spanish Ministry of Economy and Competitiveness under Grant PID2020-116683GB-C22. Francesc Molina was also supported by a Margarita Salas Fellowship.}
}

\date{\today}
%\markboth{Submitted to IEEE Transactions on Information Theory, \today}%
%{Shell \MakeLowercase{\textit{et al.}}: Bare Demo of IEEEtran.cls for Journals}

\maketitle

\begin{abstract}
This paper investigates achievable information rates and error exponents of mismatched decoding when the channel belongs to the class of channels that are close to the decoding metric in terms of relative entropy. 
For both discrete- and continuous-alphabet channels, we derive approximations of the worst-case achievable information rates and error exponents as a function of the radius of a small relative entropy ball centered at the decoding metric, allowing the characterization of the loss incurred due to imperfect channel estimation.
We provide a number of examples including symmetric metrics and modulo-additive noise metrics for discrete systems, and nearest neighbor decoding for continuous-alphabet channels, where we derive the approximation when the channel admits arbitrary statistics and when it is assumed noise-additive with unknown finite second-order moment.
\end{abstract}


%\AGiF{Add channel-centric perspective to conclusions -- use Hellinger distance bound. Upper bound on error.}

%%%%%%%%%%%%%%%%%%%%%%%%%%%%%%%%
%%%%%%%%%%%%%%%%%%%%%%%%%%%%%%%%
\section{Introduction}
%\IEEEPARstart{S}{tandard} 
Standard results in information theory have been derived assuming that the channel transition probability law is perfectly known to the system designer as well as to the transmitter and receiver \cite{GallagerWiley1968}. This assumption is too optimistic in practice since there are several environments where the channel is either unknown or difficult to estimate accurately. Therefore, devising the fundamental limits for reliable transmission under channel uncertainty becomes imperative \cite{lapidoth1998reliable,370120, Foundations}. 

Most practical receivers aim at estimating the channel law and use this estimate for decoding as if it were perfect.
Mismatched decoding (see e.g. \cite{Foundations} and references therein), studies precisely the problem of reliable communication with a fixed decoding metric not necessarily equal to the channel. Beyond channel uncertainty, mismatched decoding encompasses a number of important problems such as bit-interleaved coded modulation, finite-precision arithmetic and zero-error communication \cite{Foundations}.

This manuscript delves into channel uncertainty by introducing a  parameterization for the mismatched compound channel \cite[Ch. 2.2.4]{Foundations}. We model channel uncertainty as a compound channel class comprising all channels that are contained within a small radius relative entropy ball centered at the known decoding metric. Accurate approximations to achievable information rates and error exponents are derived using i.i.d. and constant-composition codes for discrete memoryless channels (DMC), and i.i.d. and cost-constrained codes for continuous-alphabet channels. The key tool used is the approximation of relative entropy adopted by Borade and Zheng in \cite{EuclideanInformationTheory}. This, combined with convex optimisation \cite{ConvexOptimization} and variational calculus \cite{gelfand2000calculus} tools, provides closed-form results with a strong conceptual understanding of the sensitivity to small channel uncertainties.

%%%%%%%%%%%%%%%%%%%%%%%%%%%%%%%%
\subsection{Notation}
Vectors are indicated in boldface and the corresponding $i$-th entry is written using a subscript, as $\xv$ and $x_i$. We denote the transpose of $\xv$ by $\xv^T$.
Random variables are denoted in uppercase. The expectation and the variance of random variable $X$  under probability distribution $\QX$ are denoted respectively by $\EX_{\QX}[X]$ and $\VAR_{\QX}[X]$, or simply $\EX[X]$ and $\VAR[X]$; emphasis on $\QX$ will be made when ambiguity arises.
%The input and output alphabets are $\Xc$ and $\Yc$, respectively. The alphabet from the Cartesian product of $\Xc$ and $\Yc$  is $\Xc\times\Yc$. 
The cardinality of alphabet $\Xc$ is indicated as $|\Xc|$. The relative entropy and chi-squared distance between probability distributions $P$ and $P'$ are $\DKL(P\| P')$ and $\mathcal{X}^2(P, P')$, respectively. Throughout the paper, natural logarithms are adopted.

%%%%%%%%%%%%%%%%%%%%%%%%%%%%%%%%
\subsection{Problem Setup} \label{sec:Setup}
Consider reliable transmission of $M$ messages over a DMC with input $X$ and output $Y$, taking values from discrete and finite alphabets $\Xc$ and $\Yc$, respectively.
The input distribution is denoted by $\QX(x)=\Pr[X=x]$ for all $x\in\Xc$ and the channel probability distribution is $W(y|x) = \Pr[Y=y | X=x]$ for all pairs $(x, y)\in\Xc\times\Yc$.
For transmission, the encoder selects a message $m \in \{1,\dots, M\}$ and transmits the corresponding $n$-symbol codeword $\xv^{(m)}=(x_1^{(m)}, \dotsc, x_n^{(m)})$  from the codebook $\Cc_n = \{\xv^{(i)}\}_{1\leq i \leq M}$. Upon observing $\yv\in\Yc^n$, the decoder estimates the transmitted message, inspired by the maximum-likelihood decoder but using a possibly suboptimal decoding metric $q^n(\xv, \yv)$, as%
\begin{equation}
	\widehat{m} = \argmax_{1\leq \bar{m} \leq M}\ q^n(\xv^{(\bar{m})} \!, \yv).%
\end{equation}
When $q^n(\xv,\yv) \equiv \prod_{i=1}^{n} W(y_i|x_i)$, the decoder is the optimal maximum-likelihood (ML) decoder and is commonly referred to as matched. The symbol $\equiv$ denotes equivalence in decoding metric in the sense of \cite[Propositions 2.1 and 2.2]{Foundations}. In any other case, the decoder is said to be mismatched. An error is declared when $\widehat{m} \neq m$, and the probability of error for $\Cc_n$ is defined as $p_e(\Cc_n) = \Pr[\widehat{m}\neq m]$.
In this work, we will work with product metrics for which we adopt the symbol metric notation $q(x, y)$ and
\begin{equation}
	\widehat{m} = \operatorname*{argmax}_{1\leq \bar{m} \leq M}\ \prod_{i=1}^n q(x^{(\bar{m})}_i \!, y_i).%
\end{equation}
After taking logarithms, these become additive metrics.

We consider the relevant case in practice where the decoding metric $q(x,y) = \widehat{W}(y|x)>0$ is a channel estimate $\widehat{W}(y|x)$ from the output of a channel estimator.
We analyze achievable rates and error exponents for an unknown channel $W(y|x)$ and a mismatched decoder that uses  $\widehat{W}(y|x)$ as if it were perfect, and impose a constraint on the level of mismatch between the estimated and true channels by defining an appropriate statistical distance. 
More concretely, let $d(P, P^\prime)$ be a statistical distance between the joint probability distributions $P$ and $P^\prime$ over the alphabet $\Xc {\times} \Yc$. Then we define \emph{small mismatch} as the class of channels \cite{MismatchedHypTesting} such that%
\begin{equation}\label{eq:ball}
	W \in \big\{P_{Y|X} \colon d(\QX P_{Y|X}, \QX\widehat{W}) \le r \big\}.
\end{equation}
That is, $W$ is within a ball centered at $\widehat W$ of radius $r$, assumed to be small.
This class of channels is useful for modeling small perturbations of the channel around the decoding metric due to good, yet imperfect estimation.

%%%%%%%%%%%%%%%%%%%%%%%%%%%%%%%%
\subsection{Overview of Achievable Rates}
A number of achievable rates have been derived in the literature for DMCs and continuous channels with an arbitrary fixed decoding metric $q(x,y) >0$ (see e.g. \cite{Foundations} for a recent survey).
The GMI (generalized mutual information) $I_{\text{GMI}}$ \cite{kaplan1993ira} is the achievable rate obtained with i.i.d. random codes. Instead, employing constant-composition codes for DMCs or cost-constrained codes with a single per-codeword cost constraint for arbitrary alphabets yields the LM (lower mismatch) rate $I_{\text{LM}}$ \cite{Hui83,CsiszarKorner81graph}.
In this work, we focus on the dual expressions of the aforementioned rates defined from the \emph{mismatched information density}%
\begin{equation}
\label{eqn:i_sa_def_generic}
i_{s, a}(x, y) \triangleq \log \frac{q(x,y)^s  e^{a(x)}}{\EX_{\QX} [q(X, y)^s e^{a(X)}]}
\end{equation}
where $s\ge0$ and $a(x) \in \mathbb{R}$ for $x\in\Xc$. Therefore, as \cite[Ch. 2.3.1, Ch. 2.3.2]{Foundations}%
\begin{align}
I_{\text{GMI}}(\QX) &\triangleq\ \operatorname*{sup}_{s \geq 0}\ \EX_{\QX{\times}W} \big[i_{s, 0}(X, Y)\big]  \label{eqn:GMIDualExpression}\\
I_{\text{LM}}(\QX) &\triangleq\! \operatorname*{sup}_{s \ge 0, a(\cdot)} \EX_{\QX{\times}W} \big[  i_{s,a}(X, Y) \big], \label{eqn:LMDualExpression}
\end{align}
where the notation $i_{s, 0}(x, y)$ in \eqref{eqn:GMIDualExpression} corresponds to setting $a(x) = 0$ in \eqref{eqn:i_sa_def_generic} for every $x\in\Xc$. Indeed, the GMI rate can be obtained from the LM rate by setting the function $a(x)$ to zero for every input.
The GMI and LM rates are tight with respect to their respective ensemble of codes. Thus, the weakness in rate relative to the mismatch capacity $C_{\text{M}}$, the supremum of all achievable rates, is due to an inherent weakness of employed codes rather than loose mathematical analysis. Hence $I_{\text{GMI}}(\QX) \le I_{\text{LM}}(\QX) \le C_{\text{M}}$, with equality in specific cases \cite{UpperboundEhsan, MulticastAnelia}.

%%%%%%%%%%%%%%%%%%%%%%%%%%%%%%%%
\subsection{Overview of Error Exponents}
\label{RandomCodingErrorExponents}

An error exponent $E(R)$ is said to be achievable at rate $R$ whenever
\begin{equation}
    \label{eqn:probability_of_error_bound}
    E_{\text{r}}(R) \le \operatorname*{lim}_{n \to \infty} - \frac{1}{n} \log \bar{p}_e(n, \lfloor e^{nR} \rfloor),
\end{equation}
where $\bar{p}_e$ is the average error probability over all randomly generated codes of a specific ensemble. If \eqref{eqn:probability_of_error_bound} holds with equality, $E_{\text{r}}(R)$ is said to be the ensemble-tight exponent \cite[Ch. 7.1]{Foundations}. For ML decoding,
\begin{align} \label{eqn:RCErrorExponent}
    E_{\text{r}}(\QX, R) &\triangleq \operatorname*{max}_{\rho \in [0, 1]} E_0(\QX, \rho) - \rho R
\end{align}
is the (ensemble-tight) random coding exponent \cite{GallagerWiley1968} achievable at rate $R$ with i.i.d and constant-composition codes, where the respective Gallager $E_0$ functions are defined from%
\begin{equation}
	\varepsilon_{a, \rho}(x, y) \triangleq \bigg( \frac{ \EX_{\QX}[W(y|X)^{\frac{1}{1+\rho}}e^{a(X)}]}{W(y|x)^{\frac{1}{1+\rho}} e^{a(x)}} \bigg)^{\rho}%
\end{equation}
as%
\begin{align} \label{eqn:MatchediidGallagerFunction}
E_0^{\text{iid}}(\QX, \rho) &\triangleq - \log \EX_{\QX \times W} [\varepsilon_{0, \rho}(X, Y) ]\\
\label{eqn:MatchedccGallagerFunction}
E_0^{\text{cc}}(\QX, \rho) &\triangleq \sup_{a(\cdot)} -\EX_{\QX} \big[\! \log \EX_{W} [\varepsilon_{a, \rho}(X, Y)|X ] \big]. %
\end{align}
It is known that the error exponent $E_{\text{r}}(\QX, R)$ is positive for rates $R < \frac{\partial E_0(\QX, \rho)}{\partial\rho}\big|_{\rho=0}=I(X;Y)$, proving the achievability part of the channel coding theorem. Moreover, $E_0^{\text{cc}}(\QX, \rho)\geq E_0^{\text{iid}}(\QX, \rho)$ with equality only for the respective optimal input distributions \cite[Eq. (52)]{6763080}.

Analogous to the matched case, the ensemble-tight mismatched random coding error exponents \cite[Th. 7.1]{Foundations} of i.i.d.\ and constant-composition codes are defined using%
\begin{equation}\label{eqn:e_s_a_def}
	\varepsilon_{s, a, \rho}(x, y) \triangleq  \bigg( \frac{\EX_{\QX}[q(X,y)^s e^{a(X)}]}{q(x,y)^s e^{a(x)}} \bigg) ^\rho
\end{equation}
to obtain the following dual expressions \cite[Ch. 7.2]{Foundations}:%
\begin{align}
    E_\text{r}^{\text{iid}}(\QX, R) &\triangleq \operatorname*{max}_{\rho \in [0,1]} E_0^{\text{iid}}(\QX, \rho) - \rho R \label{eqn:ErIID}\\
    E_\text{r}^{\text{cc}}(\QX, R) &\triangleq \operatorname*{max}_{\rho \in [0,1]} E_0^{\text{cc}}(\QX, \rho) - \rho R,\label{eqn:ErCC}
\end{align}
where the respective mismatched Gallager functions are%
\begin{align} 
E_0^{\text{iid}}(\QX, \rho) &\triangleq\ \sup_{s \ge 0}\ -\log \EX_{\QX{\times}W} [\varepsilon_{s, 0, \rho}(X, Y)] \label{eqn:E0IID}\\
E_0^{\text{cc}}(\QX, \rho) &\triangleq\! {\sup_{s \ge 0, a(\cdot)}} \hspace{-0.25em}{-}\EX_{\QX} \big[\! \log \EX_{W} [\varepsilon_{s, a, \rho}(X, Y)|X] \big]{.} \label{eqn:E0CC}
\end{align}
For simplicity, we have used the same notation for the mismatched error exponents as for the matched case.

An alternative expression for $E_0^{\text{cc}}$ can be formulated by observing that constant-composition and cost-constrained coding achieve the same exponent when two cost functions $c_1(\cdot)$ and $c_2(\cdot)$ are optimized \cite[Th. 4]{6763080}. This is%
\begin{align}\label{eqn:E0Cost}
	E_0^{\text{cc}}(\QX, \rho) &=\! \sup_{c_1\!(\cdot), c_2\!(\cdot)} E_0^{\text{cost}}(\QX, \rho, \{c_1, c_2\})
\end{align}
where%
\begin{align}
	E_0^{\text{cost}}(\QX, \rho, \cv) &\triangleq\! \sup_{s\ge0, \boldsymbol{\lambda}, \bar{\boldsymbol{\lambda}}} \hspace{-0.5em} -\log \EX_{\QX\times W} \big[ \varepsilon_{s,\lv, \bar{\lv}, \rho}^\text{cost}(X, Y) \big]\\
	\varepsilon_{s, \lv, \bar{\lv},\rho}^{\text{cost}}(x, y) &\triangleq \bigg( \frac{\EX_{\QX}[q(X, y)^s e^{\bar\lv^{\!T\!} (\cv(X)-\pv)}] }{q(x, y)^s e^{{\lv}^{\!T\!} (\cv(x)-\pv)}} \bigg)^\rho%
\end{align}
and $\cv(x)\triangleq[c_1(x), c_2(x)]^T$, $\lv\triangleq[\lambda_1, \lambda_2]^T$, $\bar{\lv}\triangleq[\bar{\lambda}_1, \bar{\lambda}_2]^T$ and $\boldsymbol{\phi}\triangleq[\EX[c_1(X)], \EX[c_2(X)]]^T$ defined in vector form for ease of notation.
We are more interested in the expression \eqref{eqn:E0Cost} for the mismatched $E_0$ function of constant-composition codes, as it adopts a similar form to that of i.i.d. codes, leading to tighter derivations. In addition, cost-constrained coding remains valid for continuous alphabets.
When a matched ML decoder $q(x, y) = W(y|x)$ is used, the matched Gallager $E_0$ functions \eqref{eqn:MatchediidGallagerFunction}--\eqref{eqn:MatchedccGallagerFunction}, are recovered from \eqref{eqn:E0IID}--\eqref{eqn:E0Cost} by setting $s=(1 + \rho)^{-1}$, hence recovering the matched error exponents.
Using Jensen's inequality \cite[Ch. 2]{ElementsOfInformationTheory}, it can be shown that for any pair $(W, q)$%
\begin{equation} \label{eqn:EriidLessThanErcc}
	E_\text{r}^{\text{iid}}(\QX, R)\le E_\text{r}^{\text{cc}}(\QX, R) .%
\end{equation}

%%%%%%%%%%%%%%%%%%%%%%%%%%%%%%%%
\subsection{Contributions and Paper Outline}
Existing work on mismatched decoding typically only considers channel uncertainty at the receiver for an arbitrary channel metric pair. This paper explores achievable rates and random coding error exponents under channel uncertainty for a class of channels that are close to the decoding metric, modeling the case of accurate yet imperfect channel estimation.
This scenario is described as the compound channel \cite{669134, 7464362}, where the only knowledge available is the class of channels over which the system will operate.
The novelty of this paper lies in analyzing the compound channel under a given mismatched decoding metric. We parameterize the compound channel class by all the channels that lie within a relative entropy ball of radius $r$. 
Our primary contribution is the sensitivity analysis of achievable rates and random coding error exponents to small values of $r$, capturing the effect of accurate but imperfect channel estimation. We also assess the boundary between matched and mismatched decoding, providing novel closed-form results for symmetric decoding metrics and the nearest neighbor decoder. In particular:%
\begin{enumerate}[1)\IEEElabelindent=0em \labelsep=1pt]
	\item In Section \ref{AchievableRatesUnderSmallMismatch}, we derive worst-case information rates for i.i.d. and constant-composition random codes. We derive an expansion of the worst-case information rates for small mismatch $r$ and show that the mutual information of the estimated channel cannot be surpassed. Our expansion shows second-order terms proportional to $\sqrt{r}$, inducing an achievable rate penalty that decays with an infinite slope with respect to the mutual information of the estimated channel achieved at $r=0$. Examples of symmetric and modulo-additive decoding metrics are presented, demonstrating that the worst-case channel within the ball shares the same structure as the decoding metric.
	
	\item In Section \ref{ErrorExponentsUnderSmallMismatch}, we derive worst-case error exponents for i.i.d. and constant-composition random codes. We provide the corresponding expansions with second-order penalty terms again proportional to $\sqrt{r}$. Examples for the same symmetric and modulo-additive decoding metrics used for achievable rates are provided, yielding similar results.
	
	\item In Section \ref{ContinuousChannels}, we extend previous analyses allowing for continuous input and output alphabets. We provide expansions for the worst-case achievable rates and error exponents for i.i.d. and cost-constrained random coding with multiple auxiliary costs.
	The results are particularized to Gaussian codebooks with nearest neighbor decoding \cite{532892}. In particular, we derive a second-order expansion of the worst-case achievable rates for arbitrary channels, not necessarily additive, in a small-radius relative entropy ball around the nearest neighbor decoding metric.
	We also show a significantly smaller penalty when the channels within the ball are assumed to be additive.
	The structure of the worst-case channels is also studied, highlighting the discontinuities observed when the estimated channel has unbounded support.
\end{enumerate}
Conclusions are drawn in Section \ref{Conc}. Proofs are found in the Appendices.

%%%%%%%%%%%%%%%%%%%%%%%%%%%%%%%%
%%%%%%%%%%%%%%%%%%%%%%%%%%%%%%%%
\section{Small Mismatch: A Class of Compound Channels} \label{DefSmallMismatch}
We study a class of mismatched compound channels for which DMC $W(y|x)$ is unknown but assumed to lie within a divergence ball $\mathcal{B}$ centered on $\widehat{W}(y|x)$, a conditional probability distribution that models an estimate of $W(y|x)$ from the output of a channel estimator.
The decoder uses it as the decoding metric, i.e., $q(x, y) = \widehat{W}(y|x)>0$.

\begin{definition} \label{def:small_mismatch}
For small mismatch (ball radius) between the true and estimated channels $W$ and $\widehat{W}$, we require that%
\begin{equation} \label{eqn:WorstCaseRatesConstraint}
W \in \mathcal{B}(\widehat{W}\!, r) \triangleq \{P_{Y|X} \colon \DKL(\widehat{W} \| P_{Y|X} | \QX) \le r \}
\end{equation}
holds for small $r$, where $\DKL(\widehat{W} \| P_{Y|X} | \QX)$ is the conditional relative entropy%
\begin{equation}
	\DKL(\widehat{W} \| P_{Y|X} | \QX) \triangleq \sum_{x, y} \QX(x) \widehat{W}(y|x)\! \log\! \frac{\widehat{W}(y|x)}{P_{Y|X}(y|x)}.
\end{equation}
Eq. \eqref{eqn:WorstCaseRatesConstraint} is equivalent to setting $\DKL(\QX\widehat{W} \| \QX P_{Y|X}) \leq r$.
\end{definition}
We have adopted a decoder-centric perspective in which the ball is centered around the known quantity, i.e. the estimated channel $\widehat{W}$ employed to decode.
One of the advantages of this formulation is that for sufficiently small mismatch $r$ we can rely on \cite[Eqs. (1)--(4)]{EuclideanInformationTheory} and apply a Taylor expansion  to the logarithm in the relative entropy and group the $o(\cdot)$ terms. This allows us to express the relative entropy as a function of%
\begin{equation} \label{eqn:theta_def}
	\theta(y|x) \triangleq P_{Y|X}(y|x) - \widehat{W}(y|x)%
\end{equation}
plus a term of lesser significance, as%
\begin{align}\label{eqn:DKLapproximation}
	\DKL(\widehat{W}\| P_{Y|X}|\QX) = &\ \frac{1}{2} \sum_{x,y} \QX(x) \frac{\theta(y|x)^2}{\widehat{W}(y|x)} \nonumber\\
	&+ o\bigg(\frac{1}{2}\sum_{x, y}\! \QX(x) \frac{\theta(y|x)^2}{\widehat{W}(y|x)} \bigg).
\end{align}
Observe that the first term of \eqref{eqn:DKLapproximation} is exactly the chi-squared distance $\frac12\mathcal{X}^2 ( \QX P_{Y|X}\!, \QX \widehat{W})$ as defined in \cite[Pag. 425]{BoundsMeasures}.
Vectorizing $\theta$ allows us to express \eqref{eqn:DKLapproximation} in a more convenient form to obtain closed form results.
Besides this, there are other statistical distances to consider, such as generic $f$-divergences \cite{FDivergence, FDivergence2} or the R\'enyi divergence of order $\alpha$ \cite{RenyiEntropy}. Here, we only consider relative entropy, an $f$-divergence with $f^{\prime\prime}(1) = 1$ ($f^{\prime\prime}$ stands for the second derivative of $f$). Our results can be easily extended to a generic $f$-divergence or R\'enyi entropy by substituting the $\tfrac{1}{2}$ by the constant factor $\tfrac{f^{\prime\prime}(1)}{2}$ for $f$-divergences or $\tfrac{\alpha}{2}$ for the R\'enyi divergence. 

%%%%%%%%%%%%%%%%%%%%%%%%%%%%%%%%%%%%%%%%%%%
\section{Achievable Rates} \label{AchievableRatesUnderSmallMismatch}
In this section, we derive the worst-case LM and GMI rates.
For ease of explanation, we focus on the worst-case LM rate as the corresponding GMI can be obtained straightforwardly by setting $a(x)=0$ for all $x\in\Xc$.
We rewrite the mismatched information density $i_{s, a}$ in \eqref{eqn:i_sa_def_generic} with $q(x, y)=\widehat{W}(y|x)$ as%
\begin{equation}\label{eqn:i_sa_def}
i_{s, a}(x, y) = \log \frac{\widehat{W}(y|x)^s  e^{a(x)}}{\EX_{\QX} [\widehat{W}(y|X)^s e^{a(X)}]}.%
\end{equation}
The worst-case rate problem we wish to solve is%
\begin{align}\label{eqn:WorstCaseRateFormulation}
\underline{I}_{\text{\normalfont LM}}(\QX, \widehat{W}\!, r) &= \operatorname*{min}_{W \in \mathcal{B}} \operatorname*{sup}_{s \ge 0, a(\cdot)} \EX_{\QX{\times}W} [ i_{s,a}(X, Y)]
\end{align}
where the minimization is over all valid conditional probability distributions in the relative entropy ball $\mathcal{B}$ defined in \eqref{eqn:WorstCaseRatesConstraint}. Since the true channel is unknown, problem \eqref{eqn:WorstCaseRateFormulation} finds the channel in $\mathcal{B}$ that gives the worst possible LM rate. This gives an indication of the rate loss incurred by accurate but not perfect channel estimation.

Channels $W \in \mathcal{B}$ need to be valid conditional probability distributions, i.e., 
\begin{align} \label{eqn:Positivity1}
	W(y|x) \geq 0 &\quad \text{for} \quad (x,y) \in \mathcal{X}\times \mathcal{Y}\\ \label{eqn:Positivity2}
	\sum_{y} W(y|x) = 1 &\quad \text{for} \quad x \in \mathcal{X}.
\end{align}
The constraint $W(y|x) \leq 1$ for $(x,y) \in \mathcal{X}\times \mathcal{Y}$ is guaranteed by construction when \eqref{eqn:Positivity1}--\eqref{eqn:Positivity2} are satisfied.
Next, observe that imposing \eqref{eqn:Positivity1} in problem \eqref{eqn:WorstCaseRateFormulation} is unnecessary since the relative entropy in $\eqref{eqn:WorstCaseRatesConstraint}$ implies that positivity already for a strictly positive $\widehat{W}$.

\begin{theorem}\label{theorem:WorstCaseLMRate}
    Consider a class of DMCs $W(y|x)$, a mismatched decoder based on the channel estimate $\widehat{W}(y|x)$ and a fixed input distribution $\QX(x)$ satisfying \eqref{eqn:WorstCaseRatesConstraint}. Then, for a sufficiently small $r \ge 0$, the worst-case LM and GMI rates can be expressed as
    \begin{align}\label{eqn:FinalExpressionLMSmallMismatch}
    \underline{I}&_{\text{\normalfont LM}}(\QX, \widehat{W}\!, r) = \nonumber \\
    &{\operatorname*{sup}_{s\ge0, a(\cdot)}} \bigg\{ I_{s,a}^{\text{\normalfont ML}}(\QX, \widehat{W}) - \sqrt{2r V_{s,a}(\QX, \widehat{W}) } + o(\sqrt{r}) \bigg\} \\
    \underline{I}&_{\text{\normalfont GMI}}(\QX, \widehat{W}\!, r) = \nonumber\\
    &\ \operatorname*{sup}_{s\ge0} \bigg\{ I_{s,0}^{\text{\normalfont ML}}(\QX, \widehat{W}) - \sqrt{2r V_{s,0}(\QX, \widehat{W}) } + o(\sqrt{r}) \bigg\}
    \end{align}
    where%
    \begin{align}
        \label{eqn:AchievableRateMLDecodingLMRate}
        I^\text{\normalfont ML}_{s,a}(\QX, \widehat{W}) &\triangleq \EX_{\QX{\times}\widehat{W}} \big[i_{s,a}(X, Y)\big]\\
        V_{s,a}(\QX, \widehat{W}) &\triangleq \EX_{\QX} \big[ \VAR_{\widehat{W}}[i_{s,a}(X, Y) | X] \big].%
    \end{align}
\end{theorem}
\begin{proof}
	The full proof is provided in Appendix \ref{ProofLMRate}, where we minimize the dual expression for the LM rate \eqref{eqn:WorstCaseRateFormulation} including the $o(\cdot)$ term in \eqref{eqn:DKLapproximation} as a penalty in the objective function. This gives an accurate approximation to relative entropy \cite{EuclideanInformationTheory} and an accurate approximation on $\underline{I}_{\text{LM}}$ as $r\to 0$. The worst-case GMI rate is found by setting $a(x)=0$ for $x\in\Xc$.
\end{proof}

The following comments are in order from inspection of problem \eqref{eqn:WorstCaseRateFormulation}, Theorem \ref{theorem:WorstCaseLMRate} and its proof in Appendix \ref{ProofLMRate}:%
\begin{enumerate}[1)\IEEElabelindent=0em \labelsep=1pt]
	\item \textit{The solution for a chi-squared distance ball.}
	Let the small mismatch between the true and estimated channels $W$ and $\widehat{W}$ be characterized by a small radius chi-squared distance ball as%
	\begin{equation} \label{eqn:WorstCaseRatesConstraint2}
		W\! \in \tilde{\mathcal{B}}(\widehat{W}\!, r) \triangleq \big\{  P_{Y|X} \colon \!\tfrac{1}{2}\mathcal{X}^2( \QX P_{Y|X}\!, \QX \widehat{W}) {\le} r \big\}%
	\end{equation}
	where $\tfrac{1}{2}\mathcal{X}^2 ( \QX P_{Y|X}\!, \QX \widehat{W})$ as defined in \cite[Pag. 425]{BoundsMeasures} corresponds to the first term in \eqref{eqn:DKLapproximation}.	
	Then, the worst-case LM and GMI rates for sufficiently small $r\geq 0$ are exactly%
	\begin{align}
		\underline{\tilde{I}}_{\text{\normalfont LM}}&(\QX, \widehat{W}\!, r) = \nonumber\\
		&\operatorname*{sup}_{s\ge0, a(\cdot)} \bigg\{ I_{s,a}^{\text{\normalfont ML}}(\QX, \widehat{W}) - \sqrt{2r V_{s,a}(\QX, \widehat{W}) } \bigg\} \label{eqn:approxLMRate} \\	\underline{\tilde{I}}_{\text{\normalfont GMI}}&(\QX, \widehat{W}\!, r) = \nonumber \\
		&\operatorname*{sup}_{s\ge0} \bigg\{ I_{s,0}^{\text{\normalfont ML}}(\QX, \widehat{W}) - \sqrt{2r V_{s,0}(\QX, \widehat{W})} \bigg\} \label{eqn:approxGMI}
	\end{align}
	achieved for every fixed $s, a(x)$ by the channel distribution%
	\begin{align} 
		\tilde{W}_{\!s,a}^\ast(y|x) &= \widehat{W}(y|x) \Big( 1 - \sqrt{2r} \cdot \varphi_{s, a}(x, y) \Big) \label{eqn:optTrueWLM}
	\end{align}
	with%
	\begin{equation}\label{eqn:phi_def}
		\varphi_{s, a}(x, y) \triangleq \frac{i_{s,a}(x, y) - \EX_{\widehat{W}} [ i_{s,a}(x, Y) ] }{\sqrt{V_{s, a}(\QX, \widehat{W})}}.%
	\end{equation}	
	The worst-case channels $\tilde{W}_{\text{\normalfont LM}}^\ast, \tilde{W}_{\text{\normalfont GMI}}^\ast$ are found after maximizing \eqref{eqn:approxLMRate}--\eqref{eqn:approxGMI} over $s, a(x)$.	
	The key observation here is that the difference in the worst-case achievable rate for the mismatched compound channel with relative entropy or chi-squared distance balls is merely a little-$o$ term.
	
	\item \textit{Feasibility of worst-case channels}.
	In our proof, by shifting the $o(\cdot)$ term from the constraint to the objective function as a penalty term, we derive a more convenient and tractable mathematical analysis that yields valuable insights. However, one consequence of this shift is that the worst-case channels may have negative values, though this effect is absent for sufficiently small $r$.
	An additional condition on the ball radius $r$ must be satisfied to guarantee non-negativeness. From $1 - \sqrt{2r} \cdot \varphi_{s, a}(x, y) > 0$ for all $(x,y) \in \Xc {\times} \Yc$, we have:%
	\begin{equation}\label{eqn:r_pos}
		r < \min_{(x,y) \colon \varphi_{s,a}(x, y)>0}\ \frac{1}{2 \varphi_{s,a}(x, y)^2}.%
	\end{equation}
	$\tilde{W}_{\!s,a}^\ast(y|x) \leq 1$ for $(x,y) \in \Xc{\times}\Yc$ is satisfied by construction as \eqref{eqn:r_pos} guarantees non-negativity and $\sum_y\! \tilde{W}_{\!s,a}^\ast(y|x) = 1$.
	It is easy to show that the right-hand side of \eqref{eqn:r_pos} produces a positive value. We simply need to show that $\varphi_{s,a}(x, y)$ is finite everywhere $\widehat{W}(y|x)>0$, which can be easily shown by referring to equation \eqref{eqn:phi_def}. The only situation in which this condition would not be fulfilled is if $V_{s,a}=0$, which occurs only when $\widehat{W}(y|x)=|\Yc|^{-1}$, and for which $\underline{\tilde{I}}_{\text{\normalfont LM}}=0$. Therefore, it exists a sufficiently small $r$ that fulfills \eqref{eqn:r_pos}.
	Indeed, in all practical settings we have considered, the bound \eqref{eqn:r_pos} exceeds the ball radius of interest for the optimal $s,a(x)$ corresponding to the worst-case channels.
	
	\item \textit{The solution.} The worst-case channels are always found at the border of the respective ball, i.e., when \eqref{eqn:WorstCaseRatesConstraint} or \eqref{eqn:WorstCaseRatesConstraint2} holds with equality.
	
	\item \textit{Expansion term.} The expansions yield the same penalty term proportional to $\sqrt{r}$, the square root of the ball radius. These penalties indicate that even a small mismatch causes a severe penalty on the achievable rate.
	
	\item \label{Th1:5} \textit{Achievable rate at $r=0$.} For a fixed $\widehat{W}$, $\underline{I}_{\text{\normalfont LM}}$ \eqref{eqn:FinalExpressionLMSmallMismatch} is upper bounded by the mutual information of the estimated channel $\widehat{W}$ with input $\QX$, $I_{\text{\normalfont MI}}(\QX, \widehat{W})$. This is because  $\widehat{W} \in \mathcal{B}$ regardless of any $r\geq0$, so the minimum in \eqref{eqn:WorstCaseRateFormulation} can be upper bounded by evaluating at $W=\widehat{W}$. Hence%
	\begin{align} \label{eqn:FinalExpressionLMSmallMismatch0_1}
		\underline{I}_{\text{\normalfont LM}}(\QX, \widehat{W}\!, r)  &\leq \sup_{s \ge 0, a(\cdot)} I_{s,a}^{\text{\normalfont ML}}(\QX, \widehat{W})\\ \label{eqn:FinalExpressionLMSmallMismatch0_2}
		&= I_{\text{\normalfont MI}}(\QX, \widehat{W}),
	\end{align}
	and thus, information rates surpassing $I_{\text{\normalfont MI}}(\QX, \widehat{W})$ cannot be achieved. \eqref{eqn:FinalExpressionLMSmallMismatch0_1} is tight at $r=0$, where $\mathcal{B}=\{\widehat{W}\}$ and the supremum is achieved for $s=1$ and $a(x) = 0$ for $x\in \Xc$ \cite[Sec. IV-C]{6763080}. The same conclusions can be extracted for $\underline{\tilde{I}}_{\text{\normalfont LM}}$ defined in \eqref{eqn:approxLMRate}.
	
	\item \textit{Input-maximizing distribution.} Analysis of  \eqref{eqn:FinalExpressionLMSmallMismatch} reveals that although the first term is concave in $\QX$, the second term is not, though it is comparatively much smaller than the first, especially when $r\rightarrow0$. In this regime, $s=1$ and $a(x)=0$ are the optimal values, and therefore%
	\begin{align}
		\max_{\QX} \underline{I}_{\text{\normalfont LM}} \approx I_{1,0}(\QX^\ast, \widehat{W}) - \sqrt{2rV_{1,0}(\QX^\ast, \widehat{W})}
	\end{align} 
	where $I_{1,0}(\QX, \widehat{W}) = I_{\text{\normalfont  MI}} (\QX, \widehat{W})$ and%
	\begin{align}\label{eqn:QXoptdist}
		\QX^{\ast} \triangleq \argmax_{\QX} I_{\text{\normalfont  MI}} (\QX, \widehat{W})%
	\end{align}
	is the capacity-achieving distribution for channel $\widehat{W}$, which can be computed in a complexity efficient way via the Blahut-Arimoto algorithm \cite{1054753}. The approximation becomes tighter as $r\rightarrow0$, and holds with equality only at $r=0$, in which case%
	\begin{align}
		\max_{\QX} \underline{I}_{\text{\normalfont LM}}(\QX, \widehat{W}\!, 0) = \max_{\QX} I_{\text{\normalfont  MI}} (\QX, \widehat{W}).
	\end{align}
	The same conclusions can be extracted for $\underline{\tilde{I}}_{\text{\normalfont LM}}$ in \eqref{eqn:approxLMRate}.
	
	\item \label{Th1:7}\textit{Lower bound on $\underline{\tilde{I}}_{\text{\normalfont LM}}$ and $\underline{\tilde{I}}_{\text{\normalfont GMI}}$.} The worst-case rates \eqref{eqn:approxLMRate}--\eqref{eqn:approxGMI} can be lower-bounded by setting $s=1$ and $a(x)=0$ for every input as%
	\begin{align} \label{eqn:LowerBoundWorstCaseRates1}
		\underline{\tilde{I}}_{\text{\normalfont LM}}(\QX, \widehat{W}\!, r) 
		&\ge \underline{\tilde{I}}_{\text{\normalfont GMI}}(\QX, \widehat{W}\!, r) \\
		&\ge I_{\text{\normalfont MI}}(\QX, \widehat{W}) - \sqrt{2r V_{1, 0}(\QX, \widehat{W}) }. \label{eqn:LowerBoundWorstCaseRates2}
	\end{align}
	 As $r \to 0$ the penalty term shrinks until $I_{\text{\normalfont MI}}(\QX, \widehat{W})$ is achieved. The bound is tight at $r=0$. It allows the evaluation of the gain of introducing $s\ge0$ before applying Markov's inequality in the achievability proof \cite{Foundations}.
\end{enumerate}

\subsection{Example: Symmetric $\widehat{W}$ and Equiprobable $\QX$} \label{RateEgSymChannel}
We derive the worst-case GMI for a symmetric $\widehat{W}$ and a equiprobable input distribution $\QX(x) = |\Xc|^{-1}$ for $x\in\Xc$. This combination could be used when $W$ is close to being symmetric.
Symmetry is defined as Gallager \cite[Pag. 94]{GallagerWiley1968}, i.e. when each row or column of $\widehat{W}$ is a permutation of each other row or column, respectively.
We show the final results here; the complete derivations are found in Appendix \ref{ProofDiscreteSymmetricChannel}.

The worst-case GMI in \eqref{eqn:approxGMI} can be simplified due to the symmetry of $\widehat{W}$. Specifically, the magnitudes in \eqref{eqn:approxGMI} can be expressed using the first row of $\widehat{W}$, $\widehat{W}_{\text{sym}}(y)\triangleq W(y|x_1)$, as%
\begin{align}\label{eqn:worst_gmi_additive_channel}
I_{s,0}^{\text{\normalfont ML}}(\QX, \widehat{W}) &= \log \frac{|\Xc|}{{\sum_y} \widehat{W}_{\text{sym}}(y)^s} - sH(\widehat{W}_{\text{sym}})\\ \label{eqn:worst_gmi_additive_channel_V}
V_{s,0}(\QX, \widehat{W}) &= s^2 \Big( \EX_{\widehat{W}_{\text{sym}}}[ \log^2 \widehat{W}_{\text{sym}}] - H^2(\widehat{W}_{\text{sym}}) \Big)%
\end{align}
where $H(\widehat{W}_{\text{sym}})$ is the entropy of the random variable characterized by the probability mass function $\widehat{W}_{\text{sym}}$.

Next, we present two additional insightful results.
One of the most curious findings of this case study is that the worst-case channel in the chi-squared distance ball $\tilde{\mathcal{B}}$ defined in \eqref{eqn:WorstCaseRatesConstraint2} has the same structure as $\widehat{W}$, i.e. $\tilde{W}^*_{\text{\normalfont GMI}}$ is also a symmetric DMC and is independent of $s$. This suggests two key points. Firstly, adding $a(x)$ does not offer any advantage, and thus i.i.d. and constant-composition achieve the same rate, as for output-symmetric channels and metrics \cite[Ch. 2.4.2]{Foundations}. Secondly, for a scenario where the channel is known to be symmetric but unknown, the worst-case channel in $\tilde{\mathcal{B}}$, when considering only channels with the same structure as $\widehat{W}$, leads to the same solution we obtained.

The next result arises from applying the bounds in point \ref{Th1:7}) after Theorem \ref{theorem:WorstCaseLMRate}. Specifically, the worst-case GMI rate can be lower bounded by setting $s=1$ to yield%
\begin{align}
	 \underline{\tilde{I}}_{\text{\normalfont GMI}}(\QX, \widehat{W}\!, r) \geq C({\widehat{W}}) - \sqrt{2 r \VAR_{\widehat{W}_{\text{sym}}}[\log \widehat{W}_{\text{sym}}]},%
\end{align}
where $C({\widehat{W}}) \triangleq \log |\Xc| {-} H(\widehat{W}_{\text{sym}})$ is the capacity of DMC $\widehat{W}$. The bound becomes tighter as $r\rightarrow0$.

\subsection{Example: Modulo-additive $\widehat{W}$ and Equiprobable $\QX$}
Further results can be obtained even in a more closed form for the specific case of a modulo-additive $\widehat{W}$ by particularizing  \eqref{eqn:worst_gmi_additive_channel}--\eqref{eqn:worst_gmi_additive_channel_V} for the  channel estimate%
\begin{equation} \label{eqn:additive_What}
	\widehat{W} = {\begin{bmatrix}
        \bar{p} & p & \cdots & p \\
        p & \bar{p} & \cdots & p \\
        \vdots & \vdots & \ddots & \vdots \\
        p & p & \cdots & \bar{p}
    \end{bmatrix}}
\end{equation}
with $p\in [ 0, |\Xc|^{-1} ]$ and $\bar{p} \triangleq 1- (|\Xc|{-}1)p \ge p$ defined for compactness of notation.

In this case, we have%
\begin{align}
	I_{s,0}^{\text{\normalfont ML}}(\QX, \widehat{W}) &= \log \frac{|\Xc|}{{\bar{p}}^s + (1-\bar{p})p^{s-1}} \nonumber \\
	& \hspace{3em} + s ( \bar{p}\log \bar{p} + (1-\bar{p}) \log p )\\
	V_{s,0}(\QX, \widehat{W}) &= s^2 \cdot \bar{p}(1-\bar{p}) \log^2 ( \bar{p}/p )%
\end{align}
and the worst-case channel in $\tilde{\mathcal{B}}$ is a modulo-additive DMC with crossover probabilities%
\begin{align}
	q=p \left(1+\sqrt{2r} \cdot \sqrt{ \frac{\bar{p}}{1-\bar{p}}} \right).%
\end{align}
The feasibility condition \eqref{eqn:r_pos} leads to $r < \tfrac{1}{2}\tfrac{\bar{p}}{1 - \bar{p}}$. The worst-case bound $r < \tfrac{1}{2}\tfrac{1}{|\Xc| {-} 1}$ is achieved when $p = \tfrac{1}{|\Xc|}$.

\subsection{Example: Ternary-Input Ternary-Output $\widehat{W}$}
\label{TernaryChannelRates}
\begin{figure}[t!]
	\centering
	\input{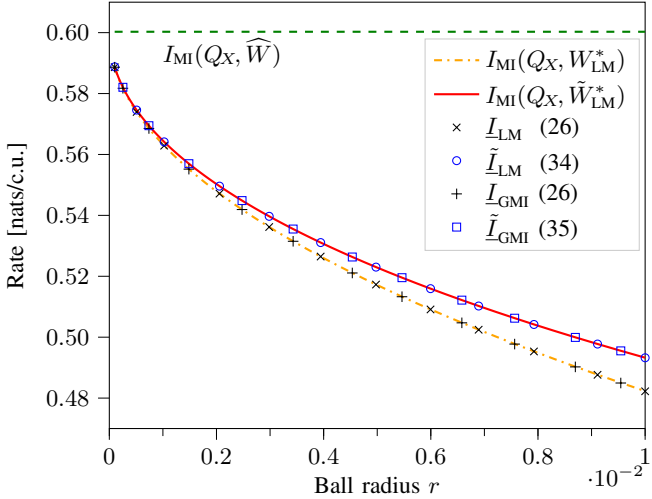}%
	\caption{Information rates (in nats per channel use) computed for $\QX$ in \eqref{eqn:TernaryInputTernaryOutputQx} and estimated channel $\widehat{W}$ in \eqref{eqn:TernaryInputTernaryOutputChannel}.}
	\label{fig:WorstCaseRates}
\end{figure}

We compute the worst-case LM and GMI rates $\underline{\tilde{I}}_{\text{LM}}, \underline{\tilde{I}}_{\text{GMI}}$ from \eqref{eqn:approxLMRate}--\eqref{eqn:approxGMI}, using the following input distribution $\QX$ and ternary-input ternary-output estimated channel $\widehat{W}$:
\begin{align} \label{eqn:TernaryInputTernaryOutputQx} \QX =& \begin{bmatrix} 0.3 & 0.3 & 0.4\end{bmatrix}\\
	\widehat{W} =& \begin{bmatrix}
           0.85 & 0.05 & 0.10 \\
           0.15 & 0.825 & 0.025 \\
           0.025 & 0.10 & 0.875 \\
         \end{bmatrix}.
         \label{eqn:TernaryInputTernaryOutputChannel}
\end{align}
We plot them in Fig. \ref{fig:WorstCaseRates} along with $\underline{I}_{\text{LM}}$ and $\underline{I}_{\text{GMI}}$ numerically computed from \eqref{eqn:WorstCaseRateFormulation} using an off-the-shelf solver.

$I_{\text{MI}}(\QX, \widehat{W})$ is shown in Fig. \ref{fig:WorstCaseRates} for reference. It is achievable as $r \to 0$. We also plot the mutual information of the worst-case channels obtained from the solver, $W^\ast_{\text{\normalfont LM}}$, and from our approximation, $\tilde{W}^\ast_{\text{\normalfont LM}}$.
The approximation is very accurate for $r$ sufficiently small, thus validating our approach for the regime of interest, i.e. small mismatch $r$.
For $r>0$ the curves decrease from the reference $I_{\text{MI}}(\QX, \widehat{W})$, as supported by point \ref{Th1:5}) after Theorem 1. Specifically, both the true and approximate worst-case rates decrease with an infinitely negative gradient at $r=0$. This behavior arises because the worst-case channel in the relative entropy ball is such that it causes the respective LM and GMI rate to decrease with an infinite slope. This indicates that even a small mismatch can have a significant impact on the achievable rates.
The feasibility of the worst-case channels \eqref{eqn:r_pos} are determined for $r < 2.35$, thus confirming the validity of our approach for these simulations.

The mutual information for the worst-case channels, illustrated by solid and dash-dotted lines, also exhibits the aforementioned behavior as the worst-case channels have a linear dependence on $\sqrt{r}$. Simulations also show that they are very close to the achievable rates we derived for the mismatched compound channel. This strongly suggests that the achievable rate for the compound channel coincides, within the limits of numerical accuracy, with the maximum rate achievable when communication is over the worst-case channel within the ball and the decoder is matched.

%%%%%%%%%%%%%%%%%%%%%%%%%%%%%%%%%%
\section{Error Exponents} \label{ErrorExponentsUnderSmallMismatch}
In this section, we study the worst-case mismatched random coding error exponents%
\begin{align}
    \underline{E}_\text{r}^{\text{iid}}(\QX, \widehat{W}\!, R, r) &= \operatorname*{min}_{W \in \mathcal{B}}\operatorname*{max}_{\rho \in [0,1]} E_0^{\text{iid}}(\QX,W\!, \rho) - \rho R\\
    \underline{E}_\text{r}^{\text{cc}}(\QX, \widehat{W}\!, R, r) &= \operatorname*{min}_{W \in \mathcal{B}}\operatorname*{max}_{\rho \in [0,1]} E_0^{\text{cc}}(\QX, W\!, \rho) - \rho R
\end{align}
where%
\begin{align}\label{eq:e0iid}
E_0^{\text{iid}}(\QX, W\!, \rho)&\triangleq\ \operatorname*{sup}_{s \ge 0} 	E_{s, \zerov, \zerov,\rho}(\QX,W)\\
E_0^{\text{cc}}(\QX, W\!, \rho)&\triangleq \operatorname*{sup}_{c_1(\cdot), c_2(\cdot)}\operatorname*{sup}_{s \ge 0, \lv, \bar{\lv}}\! E_{s, \lv, \bar{\lv},\rho}(\QX,W) \label{eq:e0cc}%
\end{align}
and%
\begin{align}
	E_{s, \lv, \bar{\lv},\rho}(\QX,W) &\triangleq - \log \EX_{\QX{\times}W} [\varepsilon_{s, \lv, \bar{\lv},\rho}(X,Y)],
\end{align}
where we have redefined the \emph{mismatched exponent density} as a function of the channel estimate as%
\begin{align}\label{eq:exponent_density}
	\varepsilon_{s, \lv, \bar{\lv},\rho}(x, y) &\triangleq \bigg( \frac{\EX_{\QX} [\widehat{W}(y|X)^s e^{\bar\lv^{\!T\!} (\cv(X)-\pv)}]}{\widehat{W}(y|x)^s e^{{\lv}^{\!T\!} (\cv(x)-\pv)}} \bigg)^\rho\!.%
\end{align}
We aim to find the worst-case mismatched random coding error exponents $\underline{E}{}_{\text{r}}^{\text{iid}}$ and $\underline{E}{}_{\text{r}}^{\text{cc}}$ subject to Definition \ref{def:small_mismatch}. The following lemma serves as a preliminary to our main theorems.

\begin{lemma}
\label{lemma:ExponentEquivalentToGallagerFunction}
Finding the worst-case mismatched random coding error exponent is equivalent to first finding its corresponding worst-case Gallager function, i.e.,
\begin{align}
\underline{E}_{\mathrm{r}}(\QX, \widehat{W}\!, R, r) = \operatorname*{max}_{\rho \in [0,1]} \operatorname*{min}_{W \in \mathcal{B}} E_0(\QX, W\!, \rho) - \rho R.
\end{align}
\end{lemma}
\begin{proof}
The worst-case error exponent $\underline{E}_{\text{r}}(\QX, \widehat{W}\!, R, r)$ is%
\begin{align}
    \underline{E}_{\text{r}} (\QX, \widehat{W}\!, R, r)
\label{eqn:ExponentMinMax}
    &= \operatorname*{min}_{W \in \mathcal{B}}  \operatorname*{max}_{\rho \in [0,1]} E_0(\QX, W\!, \rho) - \rho R \\
    &= \operatorname*{max}_{\rho \in [0,1]} \operatorname*{min}_{W \in \mathcal{B}} E_0(\QX, W\!, \rho) - \rho R,
\label{eqn:ExponentMax}
\end{align}
where, as the problem is concave in $\rho\in[0,1]$ and convex in $W\in\mathcal{B}$, the minimax theorem \cite[Th. 1]{Minimax} has been used to swap the order of the optimizations from \eqref{eqn:ExponentMinMax} to \eqref{eqn:ExponentMax}. Fixing $(\QX, R)$ and noticing that $\rho R$ is independent of $W$, the problem is then equivalent to minimizing $E_0$ prior to maximizing over $\rho\in[0,1]$. 
\end{proof}

Consequently, we focus on deriving the worst-case mismatched Gallager $E_0$ functions for i.i.d. and constant-composition codes%
\begin{align} \label{eqn:WorstCaseE0iid}
	\underline{E}{}_0^{\text{iid}}(\QX, \widehat{W}\!, \rho, r) = \operatorname*{min}_{W \in \mathcal{B}} &\ \operatorname*{sup}_{s \ge 0}\ E_{s,\zerov,\zerov,\rho}(\QX,W)\\
	\label{eqn:WorstCaseE0cc}
	\underline{E}{}_0^{\text{cc}}(\QX, \widehat{W}\!, \rho, r) = \operatorname*{min}_{W \in \mathcal{B}}&\ \operatorname*{sup}_{c_1\!(\cdot), c_2\!(\cdot)} \nonumber\\ 
	&\ \operatorname*{sup}_{s \ge 0, \lv, \bar{\lv}} E_{s, \lv, \bar{\lv},\rho}(\QX,W).%
\end{align}
We describe the results together in the following theorem. The corresponding worst-case error exponents need to be computed by maximizing over $\rho\in[0,1]$ afterwards, but they are not included herein for ease of readability.

%%%%%%%%%%%%%%%%%%%%%%%%%%%%%%%%
\begin{theorem}\label{theorem:WorstCaseE0}
     Consider a class of DMCs $W(y|x)$, a mismatched decoder based on the channel estimate $\widehat{W}(y|x)$ and a fixed input distribution $\QX(x)$ satisfying \eqref{eqn:WorstCaseRatesConstraint}. Then, for sufficiently small $r \ge 0$, the worst-case mismatched Gallager function for i.i.d. and constant-composition random coding can be expressed as%
    \begin{align}
    	\underline{E}_0^{\text{\normalfont iid}}(\QX&, \widehat{W}\!, \rho, r) = \operatorname*{sup}_{s\ge0} \bigg\{ E_{s, \zerov, \zerov, \rho}^{\text{\normalfont ML}}(\QX, \widehat{W}) \nonumber\\
    	&-\log \bigg( 1 + \sqrt{2 r V_{s, \zerov, \zerov, \rho}(\QX, \widehat{W})} \bigg) + o(\sqrt{r}) \bigg\}     \label{eqn:FinalExpressioniidExponentSmallMismatch}\\
        \underline{E}_0^{\text{\normalfont cc}}(\QX&, \widehat{W}\!, \rho, r) = \operatorname*{sup}_{c_1\!(\cdot), c_2\!(\cdot)}\operatorname*{sup}_{s \ge 0, \lv, \bar{\lv}} \bigg\{ E_{s, \lv, \bar{\lv}, \rho}^{\text{\normalfont ML}}(\QX, \widehat{W}) \nonumber\\
        &- \log \bigg(1+\sqrt{2r V_{s, \lv, \bar{\lv}, \rho}(\QX, \widehat{W}) } \bigg)+ o(\sqrt{r}) \bigg\} \label{eqn:FinalExpressionCcExponentSmallMismatch}
    \end{align}
    where%
    \begin{align}
    	\label{eqn:MLDecodingiidExponent}
    	E_{s, \lv, \bar{\lv}, \rho}^{\text{\normalfont ML}}(\QX, \widehat{W}) &\triangleq - \log \EX_{\QX{\times}\widehat{W}} [ \varepsilon_{s, \lv, \bar{\lv}, \rho}(X, Y) ]\\
    	\label{eqn:V_functionEps}
	    V_{s, \lv, \bar{\lv}, \rho}(\QX, \widehat{W}) &\triangleq \frac{\EX_{\QX} \big[ \text{\normalfont Var}_{\widehat{W}}[\varepsilon_{s, \lv, \bar{\lv}, \rho}(X, Y) | X] \big]}{\EX_{\QX{\times}\widehat{W}}^2[\varepsilon_{s, \lv, \bar{\lv},\rho}(X,Y)]}.%
    \end{align}
\end{theorem}
\begin{proof}
The proof is found in Appendix \ref{ProofiidGallagerFunction} and follows similar steps as the worst-case LM rate in Appendix \ref{ProofLMRate}.
\end{proof}

\begin{corollary}
    For sufficiently small $r$, the worst-case mismatched Gallager functions can be expanded as%
    \begin{align} 
        \underline{E}_0^{\text{\normalfont iid}}(\QX, \widehat{W}\!, \rho, r) &= \operatorname*{sup}_{s\ge0} \bigg\{ E_{s, \zerov, \zerov, \rho}^{\text{\normalfont ML}}(\QX, \widehat{W}) \nonumber\\ &- \sqrt{2 r V_{s, \zerov, \zerov, \rho}(\QX, \widehat{W})} + o(\sqrt{r}) \bigg\}     \\ \label{eqn:trueE0cc}
        \underline{E}_0^{\text{\normalfont cc}}(\QX, \widehat{W}\!, \rho, r) &= \operatorname*{sup}_{c_1\!(\cdot), c_2\!(\cdot)}\operatorname*{sup}_{s \ge 0, \lv, \bar{\lv}} \bigg\{ E_{s, \lv, \bar{\lv}, \rho}^{\text{\normalfont ML}}(\QX, \widehat{W}) \nonumber \\ &- \sqrt{2r V_{s, \lv, \bar{\lv}, \rho}(\QX, \widehat{W})} + o(\sqrt{r}) \bigg\}
    \end{align}
\end{corollary}
\begin{proof}
	The expressions come from a first-order Puiseux series expansion of \eqref{eqn:FinalExpressioniidExponentSmallMismatch}--\eqref{eqn:FinalExpressionCcExponentSmallMismatch} around $r=0$. The expansions incur an error of order $O(r)$ that is absorbed in $o(\sqrt{r})$ as the latter dominates the summation for small $r\geq 0$.
\end{proof}

%%%%%%%%%%%%%%%%%%%%%%%%%%%%%%
The following observations are in order:%
\begin{enumerate}[1)\IEEElabelindent=0em \labelsep=1pt]
	\item \textit{The solution for a chi-squared distance ball.}	The worst-case mismatched $E_0$ functions for a small mismatch given by a chi-squared distance ball as defined in \eqref{eqn:WorstCaseRatesConstraint2} are exactly%
	\begin{align}\label{eqn:approxE0iid}
		\underline{\tilde{E}}{}_0^{\text{\normalfont iid}}(\QX, \widehat{W}\!&, \rho, r) = \operatorname*{sup}_{s\ge0} \bigg\{ E_{s, \zerov, \zerov, \rho}^{\text{\normalfont ML}}(\QX, \widehat{W}) \nonumber\\
		&-\log \bigg( 1 + \sqrt{2 r V_{s, \zerov, \zerov, \rho}(\QX, \widehat{W})} \bigg) \bigg\}  \\ \label{eqn:approxE0cc}
		\underline{\tilde{E}}{}_0^{\text{\normalfont cc}}(\QX, \widehat{W}\!&, \rho, r) = \operatorname*{sup}_{c_1\!(\cdot), c_2\!(\cdot)}\operatorname*{sup}_{s \ge 0, \lv, \bar{\lv}} \bigg\{ E_{s, \lv, \bar{\lv}, \rho}^{\text{\normalfont ML}}(\QX, \widehat{W}) \nonumber\\ 
		&- \log \left(1+\sqrt{2r V_{s, \lv, \bar{\lv}, \rho}(\QX, \widehat{W}) } \right) \!\bigg\}
	\end{align}
	achieved for fixed $s, \lv, \bar{\lv}, c_1(\cdot), c_2(\cdot)$ by the channel distribution%
	\begin{align}
		\label{eqn:optTrueWE0cc}
		\tilde{W}_{\!s,\lv,\bar{\lv}}^\ast(y|x) =\ &\widehat{W}(y|x) \nonumber \\
		&\cdot \left(1 + \frac{\sqrt{2r} \cdot\varphi_{s, \lv, \bar{\lv}, \rho}(x, y)}{\EX_{\QX{\times}\widehat{W}}[\varepsilon_{s, \lv, \bar{\lv}, \rho}(X, Y)]} \right)%
	\end{align}
	with%
	\begin{equation}
		\varphi_{s, \lv, \bar{\lv}, \rho}(x, y) \triangleq \frac{\varepsilon_{s, \lv, \bar{\lv}, \rho}(x,y) - \EX_{\widehat{W}}[\varepsilon_{s, \lv, \bar{\lv}, \rho}(x,Y)]}{\sqrt{V_{s, \lv, \bar{\lv}, \rho}(\QX, \widehat{W})} }.
	\end{equation}
	The worst-case channels $\tilde{W}_{\text{\normalfont iid}}^\ast, \tilde{W}_{\text{\normalfont cc}}^\ast$ are found after maximizing \eqref{eqn:approxE0iid}--\eqref{eqn:approxE0cc} over $s, \lv, \bar{\lv}, c_1(\cdot), c_2(\cdot)$. 
	
	\item \textit{Feasibility of worst-case channels.} Similarly to what we did for the worst-case achievable rates, the following condition over the ball radius needs to be satisfied to guarantee non-negativity of the obtained distributions:%
	\begin{equation}
		r < \min_{(x,y)\colon \varphi_{s,\lv,\bar{\lv},\rho}(x, y)>0}\ \frac{1}{2 \varphi_{s,\lv,\bar{\lv},\rho}(x, y)^2}.%
	\end{equation}
	
	\item \textit{The solution.} The worst-case channels are always found at the border of the respective ball, i.e., when \eqref{eqn:WorstCaseRatesConstraint} or \eqref{eqn:WorstCaseRatesConstraint2} holds with equality.
	
	\item \textit{Expansion term.} The expansion yields a penalty term proportional to $\sqrt{r}$ at most. This indicates that even a small mismatch causes a severe penalty on the $E_0$ function and thus on the error exponent.
		
	\item \textit{Achievable $E_0$ function at $r=0$.} For a fixed $\widehat{W}$,%
	\begin{equation} 
		\underline{E}{}_0^{\text{cc}}(\QX, \widehat{W}\!, \rho, r) \leq\!\! \operatorname*{sup}_{c_1\!(\cdot), c_2\!(\cdot)}\operatorname*{sup}_{s \ge 0, \lv, \bar{\lv}}\!\! E_{s, \lv, \bar{\lv},\rho}^{\text{ML}}(\QX,\widehat{W})
	\end{equation}
	coinciding with the matched Gallager function for channel $\widehat{W}$. The proof follows because $\widehat{W} \in \mathcal{B}$ for any $r\ge0$, so the minimum in \eqref{eqn:WorstCaseE0cc} can be upper bounded by evaluating at $W=\widehat{W}$. The bound is tight at $r=0$ as $\mathcal{B}=\{\widehat{W}\}$.	
	
	\item \textit{$E_0$ function at $\rho=0$.} The worst-case LM and GMI rates cannot be obtained by differentiating the worst-case $E_0$ functions and evaluating at $\rho=0$. This is in contrast to what occurs for $r=0$, making our separate analysis of achievable rates and error exponents highly relevant.
\end{enumerate}

%%%%%%%%%%%%%%%%%%%%%%%
\subsection{Example: Symmetric $\widehat{W}$ and Equiprobable $\QX$}
As in Section \ref{RateEgSymChannel}, we derive the worst-case i.i.d.\ Gallager function for symmetric \cite[Pag. 94]{GallagerWiley1968} estimated channels $\widehat{W}$ and equiprobable input $\QX(x) = |\Xc|^{-1}$, showing only the final results; the full derivations are found in Appendix  \ref{ProofDiscreteSymmetricChannel}.

The worst-case mismatched Gallager function yields%
\begin{align}\label{eqn:worst_iid_symmetric_channel}
	\underline{\tilde{E}}{}_0^{\text{\normalfont iid}}(\QX, \widehat{W}\!, \rho, r) &=  \operatorname*{sup}_{s\ge0} \Bigg\{ \! \log \frac{|\Xc|^\rho}{\kappa_{\!s}^\rho \kappa_{1{-}s\rho}} \nonumber \\
	&- \log \bigg( 1 + \sqrt{2r} \sqrt{\frac{\kappa_{1-2s\rho}}{\kappa_{1-s\rho}^2} - 1} \bigg)  {\Bigg\}}%
\end{align}
where $\kappa_t \triangleq \sum_y \widehat{W}(y|x_1)^t$ is an auxiliary function computed from the first row of $\widehat{W}$.%

The same observations in Section \ref{RateEgSymChannel} apply here: the worst-case channel in $\tilde{\mathcal{B}}$ exhibits the same structure as $\widehat{W}$, suggesting that i.i.d and constant-composition random coding achieve the same error exponent.

%%%%%%%%%%%%%%%%%%%%%%%%%%
\subsection{Example: Modulo-Additive $\widehat{W}$ and Equiprobable $\QX$}
Following from the previous example, we now analyze the specific case of a modulo-additive $\widehat{W}$ by particularizing \eqref{eqn:worst_iid_symmetric_channel} for the transition matrix in  \eqref{eqn:additive_What}.
This gives $\kappa_t = \bar{p}^t + (|\Xc|-1)p^t= \bar{p}^t + (1-\bar{p})p^{t-1}$, and thus%
\begin{align}
\underline{\tilde{E}}{}_0^{\text{iid}}&(\QX, \widehat{W}\!, \rho, r) = \nonumber \\
&\operatorname*{sup}_{s \ge 0} \Bigg\{ \log \frac{|\Xc|^{\rho}}{( \bar{p}^{1-s\rho} + (1{-}\bar{p})p^{-s\rho} ) ( \bar{p}^s + (1{-}\bar{p})p^{s-1})^\rho} \nonumber \\
&-  \log \Bigg( 1 + \sqrt{2r} \sqrt{\frac{ \bar{p}^{1-2s\rho} + (1{-}\bar{p})p^{-2s\rho} }{( \bar{p}^{1-s\rho} + (1{-}\bar{p})p^{-s\rho} )^2} {-} 1 }\Bigg) \! \Bigg\}.
\end{align}

\subsection{Example: Ternary-Input Ternary-Output $\widehat{W}$}
\begin{figure}[t]
	\centering
	\input{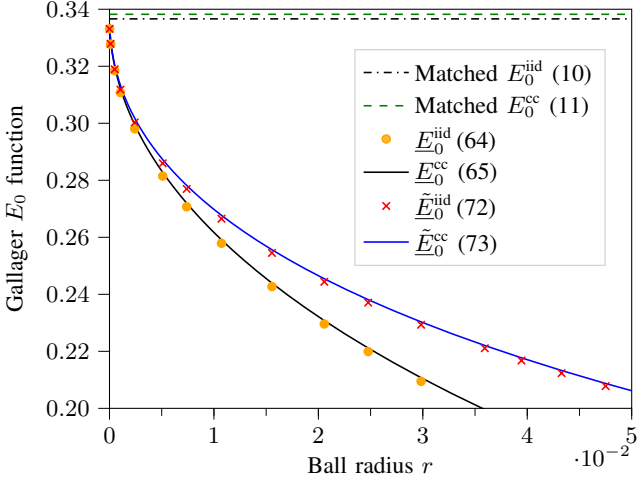}%
	\caption{Matched and mismatched Gallager functions computed for $\rho=0.7$, input distribution $\QX$ in \eqref{eqn:TernaryInputTernaryOutputQx2} and estimated channel $\widehat{W}$ in \eqref{eqn:TernaryInputTernaryOutputChannelExponents}.}
	\label{fig:AllExponents}
\end{figure}

We study the worst-case $\underline{\tilde{E}}{}_0^{\text{iid}}$ \eqref{eqn:approxE0iid} and $\underline{\tilde{E}}{}_0^{\text{cc}}$ \eqref{eqn:approxE0cc} for the following input distribution and channel estimate:%
\begin{align} \label{eqn:TernaryInputTernaryOutputQx2}
	\QX &= \begin{bmatrix} 0.3 & 0.3 & 0.4 \end{bmatrix}\\
  \widehat{W}& = \begin{bmatrix}
     0.85  & 0.05  & 0.10 \\
     0.025 & 0.945 & 0.03 \\
     0.025 & 0.10  & 0.875
 \end{bmatrix}.
         \label{eqn:TernaryInputTernaryOutputChannelExponents}
\end{align}

An off-the-shelf solver is used to numerically compute the true $\underline{E}_0^{\text{iid}}$  \eqref{eqn:WorstCaseE0iid} and $\underline{E}_0^{\text{cc}}$ \eqref{eqn:WorstCaseE0cc}, which are plotted in Fig. \ref{fig:AllExponents} along with the matched Gallager functions $E_0^{\text{iid}}$ \eqref{eqn:MatchediidGallagerFunction} and $E_0^{\text{cc}}$ \eqref{eqn:MatchedccGallagerFunction} under maximum-likelihood decoding using $\widehat{W}$.
The graph shows that $\underline{E}_0^{\text{iid}} \le \underline{E}_0^{\text{cc}}$ for all $r$, and the approximations $\underline{\tilde{E}}{}_0^{\text{\normalfont iid}}$ and $\underline{\tilde{E}}{}_0^{\text{\normalfont cc}}$ show the same behavior. Much like $I_{\text{MI}}(\QX, \widehat{W})$ in Section \ref{AchievableRatesUnderSmallMismatch}, the estimated channel's matched Gallager functions are achievable for $r=0$, but as soon as $r>0$, the mismatched Gallager functions decrease rapidly from their matched counterparts. The worst-case mismatched Gallager functions drop with infinite slope at $r=0$, as expected from their dependence on $\sqrt{r}$ in our findings.

The curves have similar graphical shapes, with approximated curves increasingly higher than the true curves. The error terms appear to be negative at the optimal point in our simulation. Apparently, with increasing $r$, we increasingly overestimate the worst-case Gallager functions, so we also overestimate the error exponents. Additionally, the simulation shows that the approximations deviate from their true values to a greater extent than those for the achievable rates. The reason is that the logarithm in \eqref{eqn:FinalExpressioniidExponentSmallMismatch}--\eqref{eqn:FinalExpressionCcExponentSmallMismatch} is more sensitive to small variations in contrast to the linear dependence in \eqref{eqn:AchievableRateMLDecodingLMRate}. Consequently, the validity of our approximation seems to be restricted to a lower range of values of $r$ for the mismatched error exponents, in contrast to the large range of values valid for the mismatched achievable rates.

%%%%%%%%%%%%%%%%%%%%%%%%%%%%%
\section{Continuous Alphabets}
\label{ContinuousChannels}
In this section, we extend the study of  worst-case achievable rates and Gallager $E_0$ functions to continuous-alphabet channels by allowing the alphabets $\Xc$ and $\Yc$ to take infinite values. Specifically, our focus is on $\Xc = \Yc = \mathbb{R}$. We use $\QX(x)$ and $W(y|x)$ to denote the continuous probability input and channel distributions, respectively, and redefine small mismatch for continuous-alphabet channels as follows. 
\begin{definition} \label{def:small_mismatch_cont}
	For small mismatch (ball radius) between the continuous-alphabet channels $W$ and $\widehat{W}$, we require that%
	\begin{align}\label{eqn:Def2}
		W \in \mathcal{B} = \{ P_{Y|X} \in \mathcal{S} \colon \DKL(\widehat{W}\|P_{Y|X}|\QX) \leq r \}
	\end{align}
	holds for small $r$, where $\mathcal{S}$ is a function space admitting  probability density functions $P_{Y|X}(y|x)$ over $(x,y) \in \mathbb{R}^2$, and $\DKL(\widehat{W}\|P_{Y|X}|\QX)$ defines the conditional relative entropy in integral form as%
	\begin{equation}
		\DKL(\widehat{W} \| P_{Y|X} | \QX) \triangleq\! {\iint} \QX(x) \widehat{W}(y|x)\! \log\! \frac{\widehat{W}(y|x)}{P_{Y|X}(y|x)} \DD x \DD y.%
	\end{equation}
\end{definition}
We analyze both i.i.d. and cost-constrained codes since constant-composition random coding is undefined for continuous alphabets.
As we shall see, most of the previous expressions of information rates and $E_0$ functions remain invariant, by substituting probability distributions with the corresponding probability densities and summations with integrals. The novelty is that the formulation of the continuous case allows for unprecedented closed-form results for the nearest neighbor decoder that further extend \cite{lapidoth1998reliable,Foundations,532892}.

%%%%%%%%%%%%%%%%%%%%%%
\subsection{i.i.d. Random Coding}
For i.i.d. random coding, the mismatched information density and mismatched exponent density are defined as
\begin{align}
	i_s(x,y) &\triangleq \log \frac{\widehat{W}(y|x)^s}{{\int} \QX(\bar{x}) \widehat{W}(y|\bar{x})^s\DD \bar{x}}\\
	\varepsilon_{s, \rho}(x, y) &\triangleq  \bigg( \frac{{\int} \QX(\bar{x}) \widehat{W}(y|\bar{x})^s \DD \bar{x}}{\widehat{W}(y|x)^s} \bigg)^\rho,%
\end{align}
where we have expressed expectations in integral form for reasons of emphasis.

We want to solve the following problems to determine the worst-case GMI and $E_0$ function:
\begin{align}\label{eqn:Cont:WorstCaseRateFormulation}
\underline{I}_{\text{\normalfont GMI}}&(\QX, \widehat{W}\!, r) = \operatorname*{min}_{W \in \mathcal{B}}\ \operatorname*{sup}_{s \ge 0}\ \EX_{\QX{\times}W} [ i_{s}(X, Y)]\\
\underline{E}{}_0^{\text{iid}}&(\QX, \widehat{W}\!, r) = \operatorname*{min}_{W \in \mathcal{B}}\ \operatorname*{sup}_{s \ge 0} -\log \EX_{\QX{\times}W} [\varepsilon_{s,\rho}(X, Y)]%
\label{eqn:Cont:WorstCaseRateFormulationExp}
\end{align}
Their solutions are particular cases of cost-constrained random coding, and are thus described below.

%%%%%%%%%%%%%%%%%%%%%%%%%%%%%%%%%
\subsection{Cost-Constrained Random Coding}
We study codebooks in which, unlike i.i.d. random coding, codewords are constrained to satisfy multiple per-codeword cost functions \cite{6763080}.
Specifically, we consider codebooks in which each codeword $\xv^{(m)}$ satisfies%
\begin{align}\label{eqn:Cont:c0}
	\frac{1}{n} \sum_{i=1}^n c_0(x_i^{(m)}) &\leq P\\\label{eqn:Cont:cl}
	\bigg\vert \frac{1}{n} \sum_{i=1}^n c_\ell(x_i^{(m)}) - \phi_\ell \bigg\vert & \leq \frac{\delta}{n} \quad, \quad \ell={1, \dots, L}%
\end{align}
where $\phi_\ell \triangleq \EX[c_\ell(X)]$ for $\ell=1,\dotsc, L$ and $\delta>0$.
Observe that these input cost functions might be related to the channel having an input cost constraint $c_0(\cdot)$, or might be design cost constraints that allow us to shape codewords in a certain way as  for the cost functions $c_{\ell}(\cdot)$ for $\ell=1,\dotsc, L$.

The corresponding mismatched information and exponent densities now read%
\begin{align}\label{eqn:Cont:FormulationCC_is}
	i_{s,\lv}(x,y) &\triangleq \log \frac{\widehat{W}(y|x)^s e^{\lv^{\!T\!} \cv(x)}}{{\int} \QX(\bar{x}) \widehat{W}(y|\bar{x})^s e^{\lv^{\!T\!}\cv(\bar{x})} \DD \bar{x}}
	\\ \label{eqn:Cont:FormulationCC_E0}
	\varepsilon_{s, \lv, \bar{\lv}, \rho}(x, y) &\triangleq  \bigg(\! \frac{{\int} \QX(\bar{x}) \widehat{W}(y|\bar{x})^s e^{\bar\lv^{\!T\!} (\cv(\bar{x})-\pv)} \DD \bar{x}}{\widehat{W}(y|x)^s e^{{\lv}^{\!T\!} (\cv(x)-\pv)}} \bigg)^{\!\rho}%
\end{align}
where we have defined the following vectors to denote  compactly the cost functions and corresponding weights:
\begin{align}
	\cv(x) &\triangleq \begin{bmatrix}
		c_0(x), c_1(x), \dotsc, c_L(x)
	\end{bmatrix}^T\\
	\pv &\triangleq \begin{bmatrix}
		P, \phi_1, \dotsc, \phi_L
	\end{bmatrix}^T\\
	\lv &\triangleq
	\begin{bmatrix}
		\lambda_0, \lambda_1, \dotsc, \lambda_L
	\end{bmatrix}^T %\quad \text{with} \quad \lambda_0\ge0, \lambda_\ell\in\RR,~ \ell=1,\dotsc,L
	\\
	\bar{\lv} &\triangleq
	\begin{bmatrix}
		\bar{\lambda}_0, \bar{\lambda}_1, \dotsc, \bar{\lambda}_L
	\end{bmatrix}^T% \quad \text{with} \quad \bar{\lambda}_0\ge0, \bar\lambda_\ell\in\RR,~ \ell=1,\dotsc,L.
\end{align}
with $\lambda_0,\bar\lambda_0\geq0$ and $\lambda_\ell,\bar\lambda_\ell\in\RR$ for $\ell=1,\dotsc,L$.
The problems we want to solve are%
\begin{align}\label{eqn:Cont:FormulationCost_is}
	\underline{I}_{\text{\normalfont cost}}(\QX&, \widehat{W}\!, r) = \operatorname*{min}_{W \in \mathcal{B}}\  \operatorname*{sup}_{s \ge 0, \lv} \EX_{\QX{\times}W} [ i_{s, \lv}(X, Y)]
	\\\label{eqn:Cont:FormulationCost_E0}
	\underline{E}{}_0^{\text{cost}}(\QX&, \widehat{W}\!, r) = \nonumber\\
	&\operatorname*{min}_{W \in \mathcal{B}} \operatorname*{sup}_{s \ge 0, \lv, \bar{\lv}} -\log \EX_{\QX{\times}W} [\varepsilon_{s,\lv,\bar{\lv},\rho}(X, Y)].%
\end{align}
The main difference with respect to the i.i.d. case is the need to further optimize over vectors $\lv,\bar{\lv}$. Note that \eqref{eqn:Cont:WorstCaseRateFormulation}--\eqref{eqn:Cont:WorstCaseRateFormulationExp} can be obtained from \eqref{eqn:Cont:FormulationCost_is}--\eqref{eqn:Cont:FormulationCost_E0} by setting $\lv=\bar{\lv}=\zerov$. The following theorem provides the solutions to all problems.

\begin{theorem}\label{theorem:CCcontinuous}
	Consider a family of continuous-alphabet channels $W(y|x)$, a mismatched decoder based on the channel estimate $\widehat{W}(y|x)$ and a fixed input distribution $\QX(x)$ satisfying \eqref{eqn:Def2}. Then, for every $r \ge 0$ sufficiently small:
	
	(i) the worst-case achievable rates for cost-constrained and i.i.d. random coding can be expressed as%
	\begin{align}
	\underline{I}_{\text{\normalfont cost}}(\QX, \widehat{W}\!, r) =\ & \operatorname*{sup}_{s\ge0, \lv} \bigg\{ I_{s, \lv}^{\text{\normalfont ML}}(\QX, \widehat{W}) \nonumber\\
	&- \sqrt{2r V_{s, \lv}(\QX, \widehat{W}) } + o(\sqrt{r}) \bigg\}\\
	\underline{I}_{\text{\normalfont GMI}}(\QX, \widehat{W}\!, r) =\ & \operatorname*{sup}_{s\ge0} \bigg\{ I_{s,\zerov}^{\text{\normalfont ML}}(\QX, \widehat{W}) \nonumber\\
	&- \sqrt{2r V_{s, \zerov}(\QX, \widehat{W})} + o(\sqrt{r}) \bigg\}
	\end{align}
	where%
	\begin{align}
		I^\text{\normalfont ML}_{s, \lv}(\QX, \widehat{W}) &\triangleq \EX_{\QX{\times}\widehat{W}} \big[i_{s,\lv}(X, Y)\big]\\
		V_{s, \lv}(\QX, \widehat{W}) &\triangleq \EX_{\QX} \big[ \text{\normalfont Var}_{\widehat{W}}[i_{s,\lv}(X, Y) | X] \big].
	\end{align}
	
	(ii) the worst-case mismatched $E_0$ functions for cost-constrained and i.i.d. random coding can be expressed as%
	\begin{align}	 \label{eqn:cont:FinalExpressionCostExponent}
	\underline{E}_0^{\text{\normalfont cost}}&(\QX, \widehat{W}\!, \rho, r) = \operatorname*{sup}_{s\ge0,\lv, \bar{\lv}} \bigg\{ E_{s,\lv,\bar{\lv}, \rho}^{\text{\normalfont ML}}(\QX, \widehat{W}) \nonumber\\
	&- \log \bigg( 1+ \sqrt{2 r V_{s,\lv,\bar{\lv},\rho}(\QX, \widehat{W})}\bigg) +o(\sqrt{r}) \bigg\}\\
	\underline{E}_0^{\text{\normalfont iid}}&(\QX, \widehat{W}\!, \rho, r) = \operatorname*{sup}_{s\ge0} \bigg\{ E_{s,\zerov,\zerov, \rho}^{\text{\normalfont ML}}(\QX, \widehat{W})\nonumber\\ 
	&- \log \bigg( 1+ \sqrt{2 r V_{s,\zerov,\zerov,\rho}(\QX, \widehat{W})} \bigg) +o(\sqrt{r}) \bigg\} 
	\label{eqn:cont:FinalExpressioniidExponent}
	\end{align}
	where%
	\begin{align}
		\label{eqn:Cont:MLDecodingiidExponent}
		E_{s, \lv, \bar{\lv}, \rho}^{\text{\normalfont ML}}(\QX, \widehat{W}) &\triangleq -\log \EX_{\QX{\times}\widehat{W}} \big[\varepsilon_{s, \lv, \bar{\lv}, \rho}(X, Y)\big]\\
		V_{s, \lv, \bar{\lv}, \rho}(\QX, \widehat{W}) &\triangleq \frac{\EX_{\QX} \big[ \text{\normalfont Var}_{\widehat{W}}[\varepsilon_{s, \lv, \bar{\lv},\rho}(X, Y) | X] \big]}{\EX_{\QX{\times}\widehat{W}}^2[\varepsilon_{s, \lv, \bar{\lv}, \rho}(X,Y)]}.
	\end{align}
%The bounds are tight as $r\rightarrow0$.
\end{theorem}
\begin{proof}
	The solution to worst-case achievable rates is found by resorting to the calculus of variations \cite{gelfand2000calculus}; see Appendix \ref{App:ProofCont} for the proof. The worst-case mismatched Gallager $E_0$ function is a reformulation of the worst-case achievable rates, as shown in Appendix \ref{ProofiidGallagerFunction}, followed by the steps in Appendix \ref{App:ProofCont}.
\end{proof}
%In the following, we discuss two examples of application of the above results to nearest neighbor decoding.%

%%%%%%%%%%%%%%%%%%%%%%%%%%
\subsection{Example: Gaussian Codebooks and Nearest Neighbor Decoding for Arbitrary Continuous-Alphabet Channels}
We next study the nearest neighbor decoder when combined with Gaussian inputs by particularizing the above results to the estimated channel $\widehat{W}(y|x)=\mathcal{N}(y{-x};\hat{\sigma}^2)$ and input distribution $\QX(x)=\mathcal{N}(x;P)$.
We make no additional assumption on the nature of the memoryless channel, other than being in a relative entropy ball of radius $r$ centered at the nearest neighbor decoder.
Appendix \ref{app:aux} details the specific computations; here we only show the final results.

\subsubsection{Gaussian i.i.d. codebooks}
the worst-case GMI rate can be expressed as%
\begin{align} \label{eqn:Cont:Igmi_opt}
	\underline{{I}}_{\text{\normalfont GMI}}(\QX,\widehat{W}\!, r) =\ & \sup_{s \ge 0} \Bigg\{ \frac{1}{2} \log\left(1+s\Gamma\right) + \frac{\Gamma(1-s)}{2(s^{-1}+\Gamma)} 
	\nonumber \\
	&- \sqrt{r \cdot \frac{\Gamma(2+s^2\Gamma)}{(s^{-1} + \Gamma)^2}} + o(\sqrt{r}) \Bigg\}%
\end{align}
where $\Gamma \triangleq P/\hat{\sigma}^2$ denotes the estimated signal-to-noise ratio assuming an additive noise structure even if the channels in $\mathcal{B}$ are not necessarily additive.

Moreover, as $r\rightarrow 0$ \eqref{eqn:Cont:Igmi_opt} is dominated by the first two terms, for which $s=1$ is the optimal value and thus \eqref{eqn:Cont:Igmi_opt} can be accurately approximated as%
\begin{align} \label{eqn:Cont:Igmi_opt_approx}
\underline{{I}}_{\text{\normalfont GMI}}(\QX,\widehat{W}\!,r) &\approx \frac{1}{2} \log\left(1+\Gamma\right) - \sqrt{r \cdot \frac{ \Gamma(2+\Gamma)}{(1+\Gamma)^2}}.%
\end{align}
Intuitively, information rates close to second order coding rates can be achieved when the true channel is close to a Gaussian channel in terms of relative entropy. 
This approximation, aided by \eqref{eqn:Cont:Igmi_opt}, helps quantify the gain of the parameter $s\ge0$. Fig. \ref{fig:cont} presents a numerical example of this gain.

A similar analysis  can be performed for the worst-case Gallager $E_0$ function (see Appendix \ref{app:exp_gauss} for details). The corresponding terms in the approximation are given by
\begin{align}
	E_{s, \zerov, \zerov, \rho}^{\text{\normalfont ML}}(\QX, \widehat{W}) &= \frac{\rho}{2}  \log(1+\Gamma s) \nonumber\\
	&+ \frac{1}{2} \log\! \left(\! 1+ \frac{\rho \Gamma (1-s-\rho s)}{\Gamma+s^{-1}}\! \right)
\end{align}
and%
\begin{align}
	V_{s,\zerov,\zerov,\rho}(\QX, \widehat{W}) &= \frac{1}{\sqrt{\Gamma{+}s^{-1}}} \frac{\Gamma{+} s^{-1} + \rho \Gamma(1{-}s{-}\rho s)}{\!\sqrt{\Gamma{+}s^{-1} + 2\rho \Gamma (1{-} s {-} 2\rho s)}} \nonumber\\
	&- \frac{\sqrt{2}}{\rho s} \frac{\Gamma{+}s^{-1} + \rho \Gamma(1{-}s{-}\rho s)}{\Gamma{+} s^{-1} + \rho \Gamma (2{-}s{-}2\rho s)}.
\end{align}

\subsubsection{Gaussian spherical codebooks}
when codewords satisfy the per-codeword power constraint \eqref{eqn:Cont:c0} with $c_0(x) = x^2$, we have the following achievable rate%
\begin{align} \label{eqn:Cont:Icost_opt}
	\underline{{I}}_{\text{\normalfont cost}}(\QX&,\widehat{W}\!, r) = \sup_{s \ge 0, \lambda \geq 0} \Bigg\{ \frac{1}{2} \log\left(g(\Gamma)+s\Gamma\right) \nonumber\\
	&+ \frac{\Gamma(1{-}s)+s^{-1}g(\Gamma)(1{-}g(\Gamma))}{2(\Gamma+s^{-1}g(\Gamma))}\nonumber \\ 
	&- \sqrt{r \cdot \frac{\Gamma(2g^2(\Gamma)+s^2\Gamma)}{(s^{-1}g(\Gamma) + \Gamma)^2}} + o(\sqrt{r}) \Bigg\}
\end{align}
where $g(\Gamma) \triangleq 1-2\lambda\Gamma$. Setting $\lambda=0$ gives that $g(\Gamma)=1$ and thus $\underline{{I}}_{\text{\normalfont cost}}=\underline{I}_{\text{GMI}}$. In general, we have that $\underline{{I}}_{\text{\normalfont cost}} \geq \underline{{I}}_{\text{\normalfont GMI}}$.
Numerical evidence of such a gain is shown in Fig. \ref{fig:cont}.

%%%%%%%%%%%%%%%%%%%%%%%%%%%%%%%
\subsection{Example: Gaussian Codebooks and Nearest Neighbor Decoding for Additive Noise Channels}
Further closed-form results can be obtained when we add the assumption that the channel adds the noise $Z$, but its probability density function $W(z)$ for $z \in \mathbb{R}$ is unknown. In this case, we redefine $\mathcal{B}$ to contain only additive noise channels $W(y|x) = W(y-x)$. That is,%
\begin{equation}
	W \in \mathcal{B} = \{ P_{Y|X} \in \mathcal{S} \colon D(\widehat{W}\|P_{Y|X})\leq r\}%
\end{equation}
where thanks to the additive structure of both channel and decoding metric, the relative entropy becomes independent of the input distribution $\QX$, as%
\begin{equation}
	\DKL(\widehat{W} \| P_{Y|X}) \triangleq \int_{\mathbb{R}} \widehat{W}(y) \log \frac{\widehat{W}(y)}{P_{Y|X}(y)} \DD y.%
\end{equation}

We next study the cases of Gaussian i.i.d. and spherical codebooks in Section \ref{ContinuousChannels:Additive:Gaussian} and Gaussian i.i.d. codebooks with a fixed cost in Section \ref{ContinuousChannels:Additive:FixedCostGaussian}. The same steps in Appendix \ref{App:ProofCont} are followed to derive all the results below.

\subsubsection{Gaussian i.i.d. and spherical codebooks} \label{ContinuousChannels:Additive:Gaussian}
For additive noise channels with a known second-order moment $\EX[Z^2]$, Lapidoth \cite{532892} showed that the rate
\begin{equation}
I_{\text{\normalfont Gauss}}(\QX, \widehat{W}) = \frac{1}{2} \log \left(1+\frac{P}{\EX[Z^2]}\right)%
\end{equation}
is achievable by the nearest neighbor decoder $\widehat{W}(y|x)=\mathcal{N}(y{-}x;1)$ regardless of the additive noise distribution provided that its second-order moment $\EX[Z^2]<\infty$ is known. 
We consider instead an available estimate for the unknown second-order moment which we decompose into mean $\hat{\mu}$ and variance $\hat{\sigma}^2$, and a uncertainty level given by the ball radius $r$.
Hence, the channel estimate is the nearest neighbor decoder $\widehat{W}(y|x) = \mathcal{N}(y{-}x{-}\hat{\mu};\hat{\sigma}^2)$.
We wish to solve%
\begin{align} \label{eqn:Cont:minIgauss}
\underline{{I}}_{\text{\normalfont Gauss}}(\QX, \widehat{W}\!, r) = \min_{W \in \mathcal{B}}\ \frac{1}{2} \log \left(1+\frac{P}{\EX[(Z{-}\hat{\mu})^2]}\right)%
\end{align}
where the dependence on the noise distribution appears only as a function of $\EX[(Z{-}\hat{\mu})^2]$.
In other words, we wish to find the noise distribution in $\mathcal{B}$ that gives the highest second-order moment with respect to $\hat{\mu}$.
By approximating the relative entropy ball as before, the problem can be solved analytically, leading to the following rate%
\begin{align} \label{eqn:Cont:Igauss_opt}
	\underline{I}_{\text{\normalfont Gauss}}&(\QX, \widehat{W}\!, r) = \frac{1}{2} \log \bigg(\! 1+ \frac{\Gamma}{1+ 2\sqrt{r}} \!{\bigg)} + o(\sqrt{r})
\end{align}
which can be expanded to facilitate comparison with \eqref{eqn:Cont:Igmi_opt_approx}:%
\begin{align} \label{eqn:Cont:Igauss_opt_p1approx}
\underline{{I}}_{\text{\normalfont Gauss}}(\QX, \widehat{W}\!, r) = \frac{1}{2} \log \left( 1+ \Gamma\right) - \frac{\Gamma \sqrt{r}}{1+\Gamma} + o(\sqrt{r}).%
\end{align}
In contrast to \cite[Th. 3.9]{Foundations}, there is no penalty for considering the non-zero mean provided that the true channel is sufficiently close to a Gaussian of mean $\hat{\mu}$ and variance $\hat{\sigma}$.
The result also gives a second-order term proportional to $\sqrt{r}$, but the scaling is advantageous compared to \eqref{eqn:Cont:Igmi_opt_approx}, where channel additivity is not assumed.
Simulations in Fig. \ref{fig:cont} support this statement since knowing the additive structure of the channel operation leads to a significant rate increase, especially at low signal-to-noise ratios.
Fig. \ref{fig:cont} also shows the true optimal values obtained by solving \eqref{eqn:Cont:minIgauss} as detailed in Section \ref{sec:Cont:Channels}. These values are found to be very close to our approximation in \eqref{eqn:Cont:Igauss_opt_p1approx} when $o(\sqrt{r})=0$, a coincidence that appears to be accidental. More importantly, both approximations in \eqref{eqn:Cont:Igauss_opt}--\eqref{eqn:Cont:Igauss_opt_p1approx} are remarkably close to the true solution.

\subsubsection{Fixed-cost Gaussian codebooks} \label{ContinuousChannels:Additive:FixedCostGaussian}
For additive noise channels with a known second-order moment $\EX[Z^2]$, Scarlett \textit{et al.} \cite[Th. 3.9]{Foundations} showed that the rate
\begin{equation}
	I_{\text{\normalfont fixed-cost}}(\QX, \widehat{W}) = \frac{1}{2} \log\! \left(1+\frac{P}{\VAR[Z]}\right)%
\end{equation}
is achievable by cost-constrained random coding with a single auxiliary cost $c_1(x)=x$. In other words, the cost function $c_1(x)=x$ is able to cancel the effect of the unknown channel mean. We wish to solve%
\begin{align}
	\underline{{I}}_{\text{\normalfont fixed-cost}}(\QX, \widehat{W}\!, r) = \min_{W \in \mathcal{B}}\ \frac{1}{2} \log\! \left(1+\frac{P}{\VAR[Z]}\right)%
\end{align}
or equivalently, find the noise distribution in $\mathcal{B}$ with the highest variance. The following rate is achieved%
\begin{align}
	\underline{{I}}_{\text{\normalfont fixed-cost}}(\QX, \widehat{W}\!, r) = \frac{1}{2} \log\! \left( \! 1+ \frac{\Gamma}{1+2\sqrt{r}} \!\right) +o(\sqrt{r})
\end{align}
coinciding with \eqref{eqn:Cont:Igauss_opt}. Interestingly, the mean estimate $\hat{\mu}$  only appears in the constraint $W\in\mathcal{B}$ but the achievable rate does not depend explicitly on it. The picture is slightly different  to what occurs in \eqref{eqn:Cont:Igauss_opt}, where the channel mean has to be accurately estimated so as to design the nearest neighbor decoder accordingly.
This reinforces the usefulness of cost-constrained coding, as a simple auxiliary cost function makes the achievable rate independent of the channel mean.

\begin{figure}[t!] \centering
	\begin{centering}
    \pgfplotstableread{Tables/Ratesvsr_J.txt}\Rates
    \pgfplotstableread{Tables/Ratesvsr_J_opt.txt}\Opt
       
    \begin{tikzpicture}[scale=0.875] 
    \begin{axis}
    [%
    width=8cm,
    height=6cm,
    scale only axis,
    xmin=0,
    xmax=1e-2,
    ymin=0.36,
    ymax=0.5,
    xtick={0,2e-3,...,1e-2},
    minor xtick={1e-3,2e-3,...,1e-2},
    ytick={0.36,0.38,...,0.5},
    minor ytick={0,1,...,12},
    y tick label style={
      /pgf/number format/.cd,
      fixed,
      fixed zerofill,
      precision=2,
      /tikz/.cd
    },
    legend columns=2,   
    ylabel={Rate [bits/c.u.]},
    y label style={at={(axis description cs:0.05,.5)},rotate=0,anchor=south},
    xlabel={Ball radius $r$},
    axis background/.style={fill=white},
    legend style={legend cell align=left, draw=white!15!black, legend pos=south west, anchor=south west},
    legend style={font=\footnotesize},
    legend entries={{$\underline{I}_{\text{\normalfont cost}}$ \eqref{eqn:Cont:Icost_opt}}, {$\underline{I}_{\text{\normalfont GMI}}$ \eqref{eqn:Cont:Igmi_opt}}, {$\underline{I}_{\text{\normalfont GMI}}$ approx. \eqref{eqn:Cont:Igmi_opt_approx}},  {$\underline{{I}}_{\text{\normalfont Gauss}}$ \eqref{eqn:Cont:Igauss_opt}}, {$\underline{{I}}_{\text{\normalfont Gauss}}$ approx. \eqref{eqn:Cont:Igauss_opt_p1approx}}, {Optimal  \eqref{eqn:Cont:minIgauss}} }
    ]
    
    % Plots.    
    \addplot[blue, solid, line width=1, mark=*, mark size=1.25, mark options=solid] table[y index=5] from \Rates;
    \addplot[red, solid, line width=1, mark=none, mark size=2] table[y index=1] from \Rates;
    \addplot[red, dashed, line width=1, mark=none, mark size=2, mark options=solid] table[y index=2] from \Rates;

    \addplot[black, solid, line width=1, mark=triangle*, mark size=1.25, mark options=solid] table[y index=3] from \Rates;
    \addplot[black, dashdotted, line width=1, mark=none, mark size=2, mark options=solid] table[y index=4] from \Rates;

    \draw (axis cs:0.5e-2,0.45) ellipse (0.20cm and 0.5cm);
	\node at (axis cs:0.675e-2,0.46) {\footnotesize Additive $\widehat{W}(y{-}x)$};
	\draw (axis cs:0.5e-2,0.415) ellipse (0.20cm and 0.5cm);
	\node at (axis cs:0.675e-2,0.423) {\footnotesize Arbitrary $\widehat{W}(y|x)$};
	
	\addplot[olive, solid, line width=1, mark=x, mark size=3, only marks] table[y index=1] from \Opt;
	    
    \end{axis}
    \end{tikzpicture}%
\end{centering}
	\caption{Worst-case GMI and cost-constraint rates versus ball radius $r$ for Gaussian codewords $\QX(x) = \mathcal{N}(x;1)$ and the nearest neighbor decoder $\widehat{W}(y|x)=\mathcal{N}(y{-}x;1)$. These computations assume that $o(\sqrt{r})=0$. The optimal values from \eqref{eqn:Cont:minIgauss} are depicted using cross markers.}
	\label{fig:cont}
\end{figure}
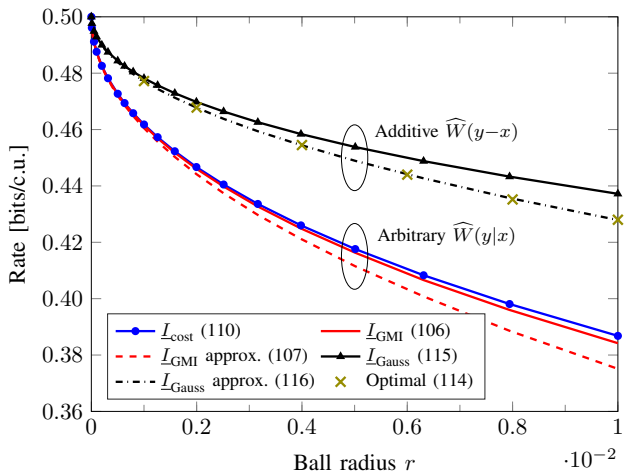

\subsection{Worst-Case Channels} \label{sec:Cont:Channels}
Throughout the manuscript, we have primarily focused on our approximation of relative entropy by chi-squared distance, which has produced the small terms $o(\sqrt{r})$ in every case we have analyzed. We have not given as much attention to the worst-case channels when examining both relative entropy and chi-squared balls.
In this respect, the continuous-alphabet case is likely one of the most intriguing cases to study as the structure of the respective worst-case channels when $\widehat{W}$ has unbounded support substantially differs depending on whether we consider a relative entropy or chi-squared distance ball.

We concentrate on the worst-case additive channel when Gaussian codebooks are combined with the nearest neighbor decoder, corresponding to%
\begin{align}
	W^\ast(z) &\triangleq \argmax_{W\in\mathcal{B}}\ \EX [Z^2] \nonumber\\
	&= \argmax_{W\in\mathcal{B}}\ \int_\mathbb{R} z^2 W(z) \DD z.
\end{align}
After tedious but straightforward analysis following standard variational tools for broken extremals \cite[Sec. 15]{gelfand2000calculus}, the worst-case channel in $\mathcal{B}$ is defined from the piece-wise expression
\begin{align}
	W^\ast_{\!z_0}(z) = \left\{ \begin{array}{lcc} \widehat{W}(z) & \text{if} & |z| \geq z_0 \\ \widehat{W}(z) \cdot \frac{\rho}{z^2-\lambda} & \text{if} & |z| < z_0 \end{array} \right.
\end{align}
for $\rho<0$ and $0<z_0<\lambda$ satisfying%
\begin{align}
	\rho &= \left( \int_{-z_0}^{z_0} \frac{\widehat{W}(z)}{z^2-\lambda} \DD z\right)^{-1}\\
	r &= \int_{-z_0}^{z_0} \widehat{W}(z) \log\! \left(\frac{z^2-\lambda}{\rho}\right)\! \DD z.
\end{align}
The worst-case channel probability density function is found after maximizing over $z_0>0$ as%
\begin{align} \label{eqn:Cont:TrueWCadd}
	 W^\ast(z) = \argmax_{z_0>0} \int_{\mathbb{R}} z^2 W^{\ast}_{\!z_0}(z) \DD z.%
\end{align}

The solution contrasts with the worst-case channel in the chi-squared distance ball $\tilde{\mathcal{B}}$ computed as in \cite[Appendix C]{10619112}
\begin{align} \label{eqn:Cont:AppWCadd}
	\tilde{W}^\ast(z) &= \argmax_{W \in \tilde{\mathcal{B}}}\ \EX[Z^2] \nonumber\\
	&= \widehat{W}(z) \left(1+ \sqrt{r} \cdot \frac{z^2 - \hat{\sigma}^2}{\hat{\sigma}^2} \right).
\end{align}
Fig. \ref{fig:Channels_cont} illustrates these differences, showing that the discontinuities at $z=\pm z_0$ for $W^*$ are not present in the corresponding approximation, $\tilde{W}^*$, computed when the chi-squared distance is instead adopted. The primary reason for this difference is that relative entropy, viewed as a functional of the channel $W$, is not continuous everywhere when the channel estimate $\widehat{W}$ has unbounded support, such as in the case of nearest neighbor decoding. In these situations, the functional admits a broken extremal as the maximizing distribution, in which only a portion of the support is optimized.

\begin{figure}[t!] \centering
	\begin{centering}
    \pgfplotstableread{Tables/Channelscont.txt}\X
    
    \begin{tikzpicture}[scale=0.875] 
    \begin{axis}
    [%
    width=8cm,
    height=6cm,
    scale only axis,
    xmin=-6,
    xmax=+6,
    ymin=-16,
    ymax=0,
    xtick={-6,-4,...,6},
    minor xtick={-6,-5,...,6},
    ytick={-16,-14,...,0},
    %minor ytick={-16,-15,...,0},
    y tick label style={
      /pgf/number format/.cd,
      fixed,
      fixed zerofill,
      precision=0,
      /tikz/.cd
    },
    ylabel={$\log W^\ast(z)$},
    y label style={at={(axis description cs:0.05,.5)},rotate=0,anchor=south},
    xlabel={$z$},
    axis background/.style={fill=white},
    legend style={legend cell align=left, at={(axis cs:-2.5,-15)},anchor=south west},
    legend style={font=\footnotesize},
    legend entries={{$r=10^{-2}$ ; $W^\ast$}, {$r=10^{-2}$ ; $\tilde{W}^\ast$}, {$r=10^{-3}$ ; $W^\ast$}, {$r=10^{-3}$ ; $\tilde{W}^\ast$}}
    ]
    
    % Plots.    
    \addplot[blue, solid, line width=1, mark=none, mark size=1.25, mark options=solid] table[x index=0, y index=1] from \X;
    \addplot[red, dashed, line width=0.5, mark=none, mark size=2] table[x index=0, y index=2] from \X;
    
    \addplot[olive, solid, line width=1, mark=none, mark size=2, mark options=solid] table[x index=3, y index=4] from \X;
    \addplot[black, dashed, line width=0.5, mark=none*, mark size=1.25, mark options=solid] table[x index=3, y index=5] from \X;
    
    \end{axis}
    \end{tikzpicture}%
\end{centering}
	\caption{Worst-case channel (in log scale) in the relative entropy ($W^\ast$ in \eqref{eqn:Cont:TrueWCadd}) and chi-squared ($\tilde{W}^\ast$ in \eqref{eqn:Cont:AppWCadd}) balls. Computations assume Gaussian codewords $\QX(x) = \mathcal{N}(x;1)$ and the nearest neighbor decoder $\widehat{W}(y|x)=\mathcal{N}(y{-}x;1)$. The worst-case GMI rates for $W^\ast$ are shown in Fig. \ref{fig:cont}.}
	\label{fig:Channels_cont}
\end{figure}
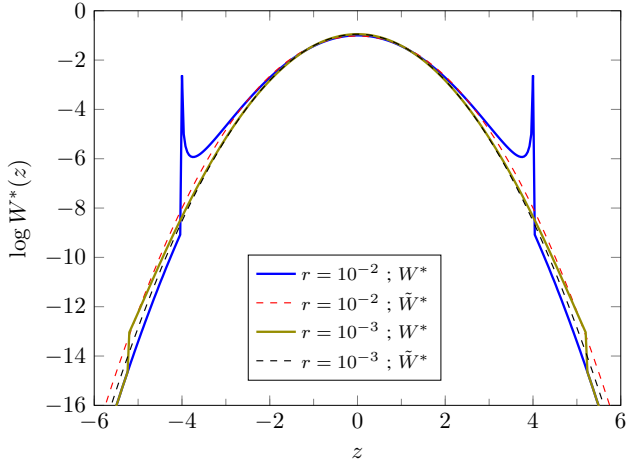

%%%%%%%%%%%%%%%%%%%%%%%%%%%%%%
\section{Conclusion} \label{Conc}
We have studied a family of mismatched compound channels by assuming that the true channel lies within a small radius relative entropy ball centered at the decoding metric. We have shown that a second order expansion of the relative entropy ball for a small radius allows the derivation of closed-form results to worst-case (minimum) information rates and error exponents for discrete and continuous memoryless channels. Specifically, for both discrete- and continuous-alphabet memoryless channels, we derive penalty terms that are proportional to the square root of the ball radius. We have obtained new explicit expressions for the worst-case rate achievable with Gaussian codebooks and nearest neighbor decoding which complement the analytical findings derived by Lapidoth and Scarlett {\em et. al} \cite{lapidoth1998reliable,340469,Foundations,532892}.
%%%%%%%%%%%%%%%%%%%%%%%%%%%%%%

\appendices
\renewcommand{\thesubsection}{\thesection-\Roman{subsection}}
\renewcommand\thesubsectiondis{\Roman{subsection}.}

\section{Proof of Theorem \ref{theorem:WorstCaseLMRate}}
\label{ProofLMRate}

Here, we provide the full derivation for the worst-case LM rate in Theorem \ref{theorem:WorstCaseLMRate}. Let us define%
\begin{align} \label{Appth1:Isa_def}
    I_{s,a}(\QX, W) &\triangleq \sum_{x,y} \QX(x) W(y|x) i_{s, a}(x, y)\\ %
    I^\text{ML}_{s,a}(\QX, \widehat{W}) &\triangleq \sum_{x,y} \QX(x) \widehat{W}(y|x) i_{s, a}(x, y), \label{Appth1:IsaML}
\end{align}
and let small mismatch be characterized by Definition \ref{def:small_mismatch}. Then, we formulate the worst-case LM rate $\underline{I}_{\text{LM}}$ as%
\begin{align} \label{Appth1:MinLMrate}
\underline{I}_{\text{LM}}(\QX, \widehat{W}\!, r) &= \operatorname*{min}_{W \in \mathcal{B}} \operatorname*{sup}_{s \ge 0, a(\cdot)} I_{s,a}(\QX, W).
\end{align}
Since $I_{s,a}$ is linear in $W$ and concave in $s\geq 0$ and $a(x)$ \cite[Ch. 2.3.4]{Foundations}, we can switch the order of the optimizations in application of the minimax theorem \cite[Th. 1]{Minimax}.
Then, by minimizing $I_{s,a}$ over all channels $W \in \mathcal{B}$ we obtain%
\begin{align}\label{Appth1:MinIsa1}
	\underline{I}_{s,a}(\QX, \widehat{W}\!, r) &\triangleq \operatorname*{min}_{W \in \mathcal{B}}\ I_{s,a}(\QX, W) \\
	\label{Appth1:MinIsa2}
	&= I^\text{ML}_{s,a}(\QX, \widehat{W}) + \operatorname*{min}_{W \in \mathcal{B}} \thetav^{T} {\nabla\! I_{s,a}}
\end{align}
where:
\begin{equation} \label{Appth1:Theta_def}
\thetav \triangleq \Big[
\theta(y_1 | x_1), \dots, \theta(y_{|\Yc|} | x_1), \theta(y_{1} | x_2), \dots, \theta(y_{|\Yc|} | x_{|\Xc|}) \Big]^T
\end{equation}
with $\theta(y_i|x_j) = W(y_i|x_j) {-} \widehat{W}(y_i|x_j)$ as defined in \eqref{eqn:theta_def}, and%
\begin{align} \label{Appth1:IsaGrad}
\nabla\! I_{s,a} \triangleq \big[ &\QX(x_1)i_{s,a}(x_1, y_1),\ \dots,\ \QX(x_1)i_{s,a}(x_1, y_{|\Yc|}), \notag \\
&\QX(x_2)i_{s,a}(x_2, y_1),\ \dots,\ \QX(x_2)i_{s,a}(x_2, y_{|\Yc|}), \notag \\
& \dots,\ \QX(x_{|\Xc|})i_{s,a}(x_{|\Xc|}, y_{|\Yc|}) \big]^T
\end{align}
is the gradient of $I_{s, a}$ \eqref{Appth1:Isa_def} with respect to each $W(y_i|x_j)$.

The constraint $W \in \mathcal{B}$ over all valid conditional probability distributions is also written in vector form as the following set of $1{+}|\Xc|$ constraints%
\begin{align}%
	\left\{ \begin{array}{rll} \label{Appth1:Constraints}
		\DKL (\widehat{W} \| W | \QX) &=& d(\thetav) {+} o(d(\thetav)) \leq r\\ %\label{Appth1:Wconstraint}
	\thetav^{T\!} \underline{\boldsymbol{1}}_j &=& 0 \quad \text{for} \quad 1 \leq j \leq |\Xc|
	\end{array} \right.	
\end{align}
under the following definitions:
\begin{align} \label{Appth1:d_def}
	d(\thetav) &\triangleq \tfrac{1}{2} \thetav^{T\!} \mathbf{K} \thetav
	\\ \label{Appth1:KMatrix_def}
	\mathbf{K} &\triangleq \text{diag}\Bigg( \frac{\QX(x_1)}{\widehat{W}(y_1|x_1)},\ \dots,\ \frac{\QX(x_{|\Xc|})}{\widehat{W}(y_{|\Yc|} | x_{|\Xc|})}\Bigg)\\	
	\underline{\boldsymbol{1}}_j &\triangleq \begin{bmatrix}
		0\ \dots\ 0\ 1_{(1,j)}\ \dots\ 1_{(|\Yc|, j)}\ 0\ \dots\ 0
	\end{bmatrix}^T
	\label{Appth1:slicedOneVectorLM}
\end{align}
where the upper side of \eqref{Appth1:Constraints} guarantees that sought channels are within the ball of radius $r$, and the constraints in the lower side guarantee that ${\sum_y} W(y|x_j)=1$ for every $x_j \in \Xc$.

Therefore, problem in \eqref{Appth1:MinIsa2} can be written equivalently as%
\begin{align}\label{Appth1:MinIsa3}
	\operatorname*{min}_{W \in \mathcal{B}} \thetav^{T} \nabla\! I_{s,a}\ = 	\operatorname*{min}_{
		\substack{ d(\thetav) + o(d(\thetav)) \le r \\ \thetav^{T\!} \underline{\boldsymbol{1}}_j = 0,\ 1\le j \le |\Xc| }} \thetav^T \nabla\! I_{s,a}\ ,%
\end{align}
a convex problem that can be solved using the Lagrangian method \cite{ConvexOptimization}. This leads to a system of implicit equations that must be solved numerically using a trial-and-error algorithm for the unknown variables.

One of the key points to note is that there is no solution to \eqref{Appth1:MinIsa3} that satisfies $\DKL(\widehat{W}\| W| \QX) < r$, with strict inequality. As the objective function and the equality constraints are linear in $\thetav$, the following minimum%
\begin{align}
 	\operatorname*{min}_{
		\thetav^{T\!} \underline{\boldsymbol{1}}_j = 0,\ 1\le j \le |\Xc| }\ \thetav^T \nabla\! I_{s,a}%
\end{align}
is unbounded unless we impose boundary constraints on $W$. In this case, the minimum is achieved by the worst-case channel%
\begin{align}
	W^\ast(y|x) = \mathds{1}\{ i_{s,a}(x,y) = \min_{\bar{y}} i_{s,a}(x,\bar{y}) \}
\end{align}
for which the constraint $\DKL(\widehat{W}\| W^\ast| \QX) = \infty \leq r$ is never satisfied. Therefore, the inequality in \eqref{Appth1:Constraints} needs to be set with equality, and we will henceforth consider it with equality.
The geometric interpretation is that, because the objective function is linear, the optimal solution is found at the boundaries of the ellipsoid defined by the constraint. Specifically, the objective function intersects the ellipsoid at two points, one of which corresponds to the minimum we are seeking.

As we have said before, the solution to \eqref{Appth1:MinIsa3} is not explicit.
To achieve results that allow for conceptual understanding, we derive an approximation that is highly accurate as $r\rightarrow0$:%
\begin{enumerate}
	\item We rely on the fact that $o(d(\boldsymbol{\theta}))$ has a negligible impact for sufficiently small $r$. This is supported by \cite{EuclideanInformationTheory}, where it is stated that $d(\boldsymbol{\theta})$ well approximates $\DKL(\widehat{W}\|W|\QX)$ when $\widehat{W}$ is in a neighborhood of $W$. Therefore, we treat $o(d(\thetav))=\gamma$ as a small constant term independent of $\thetav$, absorbed in the constraint value via $\bar{r}\triangleq r-\gamma$. We then solve the following convex problem \cite{ConvexOptimization}%
	\begin{align}\label{Appth1:Lprobl1}
		J(\bar{r})\ &= 	\operatorname*{min}_{
			\substack{ d(\thetav) = \bar{r} \\ \thetav^{T\!} \underline{\boldsymbol{1}}_j = 0,\ 1\le j \le |\Xc| }} \!\thetav^T \nabla\! I_{s,a}\\%
		&= \operatorname*{max}_{\rho, \nu_{1}, \dots, \nu_{|\Xc|}} \operatorname*{min}_{
			\substack{ \boldsymbol{\theta}}}\ \mathcal{L}(\bar{r}) \label{Appth1:Lprobl2}
	\end{align}
	where $\mathcal{L}(\bar{r})$ is the following Lagrangian%
	\begin{align}
		\mathcal{L}(\bar{r}) &\triangleq \thetav^{T} \nabla\! I_{s,a} + \rho \big(d(\boldsymbol{\theta}) - \bar{r}\big) + {\sum_j} \nu_j \thetav^{T\!} \underline{\boldsymbol{1}}_j.
	\end{align}

	\item We expand $J(\bar{r})=J(r-\gamma)$ for small $\gamma$ as%
	\begin{align}
		J(\bar{r}) &= J(r) - J^{\prime}(r) \gamma
	\end{align}
	where the first term coincides with the minimum in \eqref{Appth1:Lprobl2} when $\bar{r}=r$ given by%
	\begin{align} \label{Appth1:Lprobl3}
		J(r) &= \operatorname*{max}_{\rho, \nu_{1}, \dots, \nu_{|\Xc|}} \operatorname*{min}_{
			\substack{ \boldsymbol{\theta}}}\ \mathcal{L}(r),
	\end{align}
	and the derivative in the second term is precisely the optimal Lagrange multiplier associated with the constraint $d(\thetav)= r$ in problem \eqref{Appth1:Lprobl3}: $\rho^\ast=-J^{\prime}(r)$.
\end{enumerate}
Our approach yields $J(\bar{r})=J(r)+\rho^\ast o(d(\thetav))$, which, when evaluated at the optimal $\thetav^\ast$, produces $J(\bar{r})=J(r)+\rho^\ast o(r)$.	

Therefore, our approach establishes the following equivalence between the optimization problems when adopting the relative entropy ball $\mathcal{B}$ in \eqref{eqn:WorstCaseRatesConstraint} or the chi-squared distance ball $\tilde{\mathcal{B}}$ in \eqref{eqn:WorstCaseRatesConstraint2} for sufficiently small $r\ge0$%
\begin{align}\label{Appth1:Lprobl4}
	\operatorname*{min}_{W \in \mathcal{B}} \thetav^{T} \nabla\! I_{s,a} &= \operatorname*{min}_{W \in \tilde{\mathcal{B}}} \thetav^{T} \nabla\! I_{s,a} + \rho^{\!\ast} o(r)
\end{align}
where
\begin{align}\label{Appth1:Lprobl5}
	\operatorname*{min}_{W \in \tilde{\mathcal{B}}} \thetav^{T} \nabla\! I_{s,a}\ &\triangleq \operatorname*{max}_{\rho, \nu_{1}, \dots, \nu_{|\Xc|}} \operatorname*{min}_{
		\substack{ \boldsymbol{\theta}}}\ \mathcal{L}(r).
\end{align}

\subsection{The Solution} \label{ProofLMRate_sol}
Problem \eqref{Appth1:Lprobl5} is linear with quadratic constraints, for which the KKT conditions are necessary and sufficient for optimality \cite{ConvexOptimization}. The stationary point equations are obtained by differentiating with respect to $\thetav^{T\!}$, $\rho$, and each $\nu_j$, and equating to zero. This gives the following system of equations:%
\begin{align} \label{Appth1:ThetaOpt}
	\thetav^\ast &= - \frac{\mathbf{K}^{-1}}{\rho^{\!\ast}} \bigg(\nabla\! I_{s,a} + {\sum_j} \nu_j^\ast \underline{\boldsymbol{1}}_j\bigg)\\
	\tfrac{1}{2} {\thetav^\ast}^{T\!} \mathbf{K}\thetav^\ast &= r\\
	{\thetav^\ast}^{T\!} \underline{\boldsymbol{1}}_j &=0 \quad \text{for} \quad j=1,\dots,|\Xc|.
\end{align}
We obtain%
\begin{align} \label{Appth1:nu_j_lm}
    \nu_j^\ast   &= - \QX(x_j)\EX_{\widehat{W}} [i_{s,a}(x_j, Y)]\\
    \rho^{\!\ast} &= \sqrt{V_{s,a}(\QX, \widehat{W})/(2r)} \\
    V_{s,a}(\QX, \widehat{W}) &\triangleq \EX_{\QX}[\text{Var}_{\widehat{W}}[i_{s,a}(X, Y)|X]]%
\end{align}
and%
\begin{equation} \label{eqn:FinalThetaWithNuExpandedLM}
    \thetav^\ast = \frac{\mathbf{K}^{-1}\!}{\rho^{\ast}} \bigg( {\sum_j} \QX(x_j) \EX_{\widehat{W}} [i_{s,a}(x_j, Y)] \underline{\boldsymbol{1}}_j - \nabla\! I_{s,a} \!\bigg)%
\end{equation}
where the element $\theta^\ast(y|x)$ is extracted from $\thetav^\ast$ \eqref{eqn:FinalThetaWithNuExpandedLM} as%
\begin{align}\label{eqn:Optimal_Theta_ij}
	\theta^\ast(y|x) &= - \sqrt{2r} \cdot \widehat{W}(y | x) \varphi_{s,a}(x, y), \\ \label{App:eqn:phi_def}
	\varphi_{s,a}(x, y) &\triangleq \frac{i_{s,a}(x, y) - \EX_{\widehat{W}} [ i_{s,a}(x, Y) ] }{\sqrt{V_{s,a}(\QX, \widehat{W})}}.
\end{align}
The worst-case channel distribution obtained by isolating $\tilde{W}^{\ast}_{\!s,a}$ from $\theta^\ast(y|x) = \tilde{W}^{\ast}_{\!s,a}(y|x)-\widehat{W}(y|x)$ is%
\begin{align} 
	\tilde{W}_{\!s,a}^\ast(y|x) &= \widehat{W}(y|x) \Big( 1 - \sqrt{2r} \cdot \varphi_{s,a}(x, y) \Big).
\end{align}
We first substitute \eqref{eqn:Optimal_Theta_ij} into \eqref{Appth1:Lprobl4}, and account for the term $\rho^{\!\ast} o(r) = o(\sqrt{r})$. We then substitute the result into \eqref{Appth1:MinIsa2}, yielding the following expression:%
\begin{align}
    \underline{I}_{s,a}(\QX, \widehat{W}\!, r) =\ & I_{s,a}^{\text{ML}}(\QX, \widehat{W}) \nonumber\\
    &- \sqrt{2r V_{s,a}(\QX, \widehat{W}) } + o(\sqrt{r}).%
\end{align}
It is straightforward to check that the solution corresponds to a minimum as the cost function is linear with quadratic constraints and $\rho^{\!\ast}>0$.

\subsection{Discussion} \label{ProofLMRate_disc}
The only remaining issue to be proven is whether the worst-case channels are non-negative. This arises because the constraint $d(\thetav) = r$ in \eqref{Appth1:Lprobl5} does not necessarily guarantee non-negativity, unlike the case with $\DKL(\widehat{W} \| W | \QX) = d(\thetav) + o(d(\thetav)) = r$ for $\widehat{W} > 0$ in the original problem \eqref{Appth1:MinIsa3}.
Further analysis reveals that this issue is fully resolved when $r$ is small, without the need to impose additional constraints.

To evidence the above, we have solved the full optimization problem including constraints to ensure that all valid conditional channel distributions are non-negative, as $W(y|x) \ge 0$ for $(x,y)\in\Xc\times\Yc$, through the slack variables $\mu(x, y)\ge0$ satisfying $\mu(x,y) W(y|x)=0$. 
The solution can be expressed in terms of the following variables%
\begin{align}
	i_{s,a}^\mu(x,y) &\triangleq i_{s,a}(x,y) - \mu(x, y)\\
	\varphi_{s,a}^\mu(y|x) &\triangleq \frac{i_{s,a}^{\mu}(x, y) - \EX_{\widehat{W}}[i_{s,a}^{\mu}(x, Y)]}{ \sqrt{V_{s,a}^{\mu}(\QX, \widehat{W})} } \\
	V_{s,a}^\mu(\QX, \widehat{W}) &\triangleq \EX_{\QX}[\VAR_{\widehat{W}}[i_{s,a}^\mu(X, Y)|X]],
\end{align}
achieving the following minimum
\begin{align}
	\underline{I}_{s,a}(\QX, \widehat{W}\!, r) =\ & \max_{\mu(\cdot)\ge0} \bigg\{  \EX_{\QX{\times}\widehat{W}}[i_{s,a}^{\mu}(X, Y)]  \nonumber\\
	&- \sqrt{2r V_{s,a}^\mu(\QX, \widehat{W})} + o(\sqrt{r}) \bigg\}%
\end{align}
for the worst-case channel distribution%
\begin{align}
	\tilde{W}_{\!s,a}^\ast(y|x) &= \widehat{W}(y|x) \Big( 1 - \sqrt{2r} \cdot \varphi_{s,a}^\mu(x, y) \Big).
\end{align}
From the above expression, it is evident that when $r$ is sufficiently small, the second term has a minimal impact on the structure of $\widehat{W}(y|x)>0$.
This confirms the validity of our approach in the absence of positivity constraints for the regime of $r$ values that make $\tilde{W}_{\!s,a}^\ast(y|x) \geq 0$ when $\mu(x,y) = 0$.

\section{Proof of Theorem \ref{theorem:WorstCaseE0}}
\label{ProofiidGallagerFunction}
We determine the worst-case Gallager $E_0$ functions
\begin{align} \label{Appth2:WorstCaseE0iid}
	\underline{E}{}_0^{\text{iid}}(\QX, \widehat{W}\!, \rho, r) &= \operatorname*{min}_{W \in \mathcal{B}}\ \operatorname*{sup}_{s \ge 0}\ E_{s,\zerov,\zerov,\rho}(\QX,W)\\
	\label{Appth2:WorstCaseE0cc}
	\underline{E}{}_0^{\text{cc}}(\QX, \widehat{W}\!, \rho, r) &= \operatorname*{min}_{W \in \mathcal{B}} \operatorname*{sup}_{c_1\!(\cdot), c_2\!(\cdot)}\operatorname*{sup}_{s \ge 0, \lv, \bar{\lv}} \! E_{s, \lv, \bar{\lv},\rho}(\QX,W)%
\end{align}
using analogous arguments to those for the proof of the worst-case LM rate in Appendix \ref{ProofLMRate}. In particular, as $\underline{E}{}_0^{\text{iid}}$ can be obtained from $\underline{E}{}_0^{\text{cc}}$ by setting $\lv=\bar{\lv}=\zerov$ and by deleting the optimization over $c_1(\cdot), c_2(\cdot)$, we will only solve \eqref{Appth2:WorstCaseE0cc}. 

Firstly, as the problem is concave in $s\ge0, c_1(\cdot), c_2(\cdot)$ and convex in $W\in\mathcal{B}$, we make use of Fan's minimax Theorem \cite{Minimax} to move the minimization as the inner problem. Hence, our approach reduces to solving%
\begin{align} 
	\underline{E}_{s, \lv, \bar{\lv}, \rho}(\QX&, \widehat{W}\!, r) \triangleq \nonumber\\
	&\operatorname*{min}_{W \in \mathcal{B}} -\log \EX_{\QX{\times}W} [ \varepsilon_{s, \lv, \bar{\lv}, \rho}(X, Y)].%
	\label{eqn:logmax_E_s_rho}
\end{align}
Now, by adding and subtracting the term $E_{s, \lv, \bar{\lv}, \rho}^{\text{\normalfont ML}}(\QX, \widehat{W}) \triangleq - \log \EX_{\QX{\times}\widehat{W}} [ \varepsilon_{s, \lv, \bar{\lv}, \rho}(X, Y) ]$ and using $\theta(y|x)=W(y|x)-\widehat{W}(y|x)$, problem  \eqref{eqn:logmax_E_s_rho} can be rewritten as%
\begin{align}
	\underline{E}&_{s, \lv, \bar{\lv}, \rho}(\QX, \widehat{W}\!,r) = E_{s, \lv, \bar{\lv}, \rho}^{\text{\normalfont ML}}(\QX, \widehat{W})\nonumber\\ &- \log \left( 1+\operatorname*{max}_{W \in \mathcal{B}} \sum_{x, y}\frac{ \QX(x) \theta(y|x) \varepsilon_{s, \lv, \bar{\lv}, \rho}(x, y)}{\EX_{\QX{\times}\widehat{W}} [ \varepsilon_{s, \lv, \bar{\lv}, \rho}(X, Y) ]} \right)\!.%
\end{align}
Note that the inner maximization problem,%
\begin{align}
	\operatorname*{max}_{W \in \mathcal{B}}\ \sum_{x, y}\frac{ \QX(x) \theta(y|x) \varepsilon_{s, \lv, \bar{\lv}, \rho}(x, y)}{\EX_{\QX{\times}\widehat{W}} [ \varepsilon_{s, \lv, \bar{\lv}, \rho}(X, Y) ]},%
\end{align}
has the same form as problem \eqref{Appth1:MinIsa3} in Appendix \ref{ProofLMRate}, but with a substitution of $\varepsilon_{s, \lv, \bar{\lv}, \rho}(x, y)/\EX_{\QX{\times}\widehat{W}} [ \varepsilon_{s, \lv, \bar{\lv}, \rho}(X, Y)]$ by $i_{s, a}(x, y)$ and a change from minimization to maximization.

The solution is obtained straightforwardly by mimicking the derivation in Appendix \ref{ProofLMRate}. This gives%
\begin{align}
	\theta^\ast(y|x) &= \sqrt{2r} \cdot \frac{\widehat{W}(y|x) \varphi_{s, \lv, \bar{\lv}, \rho}(x, y)}{\EX_{\QX{\times}\widehat{W}}[\varepsilon_{s, \lv, \bar{\lv}, \rho}(X, Y)]}%
\end{align}
with%
\begin{align}
	\varphi_{s, \lv, \bar{\lv}, \rho}(x, y) &\triangleq \frac{\varepsilon_{s, \lv, \bar{\lv}, \rho}(x, y) {-} \EX_{\widehat{W}} [ \varepsilon_{s, \lv, \bar{\lv}, \rho}(x, Y) ] }{\sqrt{V_{s, \lv, \bar{\lv}, \rho}(\QX, \widehat{W})}}\\
	\!\! V_{s, \lv, \bar{\lv}, \rho}(\QX, \widehat{W}) &\triangleq \frac{\EX_{\QX} \big[ \text{\normalfont Var}_{\widehat{W}}[\varepsilon_{s, \lv, \bar{\lv}, \rho}(X, Y) | X] \big]}{\EX_{\QX{\times}\widehat{W}}^2[\varepsilon_{s, \lv, \bar{\lv}, \rho}(X,Y)]}%
\end{align}
achieving the following objective%
\begin{align}
		\underline{E}_{s, \lv, \bar{\lv}, \rho}(\QX, \widehat{W}\!, r) &=  E_{s, \lv, \bar{\lv}, \rho}^{\text{\normalfont ML}}(\QX, \widehat{W}) \nonumber\\
		-\log \bigg( 1 + &\sqrt{2 r V_{s, \lv, \bar{\lv}, \rho}(\QX, \widehat{W})} + o(\sqrt{r}) \bigg) \\
	&= E_{s, \lv, \bar{\lv}, \rho}^{\text{\normalfont ML}}(\QX, \widehat{W}) \nonumber\\
	-\log \bigg( 1 + &\sqrt{2 r V_{s, \lv, \bar{\lv}, \rho}(\QX, \widehat{W})}  \bigg) + o(\sqrt{r})
\end{align}
where for the second inequality we have used that $\log(x+o(\sqrt{r})) = \log(x) + o(\sqrt{r})$.

\section{The Case of a Symmetric $\widehat{W}$}
\label{ProofDiscreteSymmetricChannel}
For a symmetric (as Gallager \cite[Pag. 94]{GallagerWiley1968}) estimated channel $\widehat{W}$ and an equiprobable input distribution $\QX(x) = |\Xc|^{-1}$, expressions can be obtained in terms of%
\begin{align}
	\kappa_{t} &\triangleq \sum_x \widehat{W}(y|x)^t = \sum_y \widehat{W}_\text{sym}(y)^t \label{App:eqn:kappa}
\end{align}
where $\widehat{W}_{\text{sym}}(y)\equiv \widehat{W}_{\text{sym}}(y|x_1)$ denotes the first row of $\widehat{W}$. The second equality in \eqref{App:eqn:kappa} holds due to the symmetry of $\widehat{W}$.

\subsection{Computations relative to the achievable rate}
The mismatched information density reads%
\begin{align}
	i_{s, 0}(x, y) &= \log \frac{\widehat{W}(y|x)^s}{ \frac{1}{|\Xc|} {\sum_{\bar{x}}} \widehat{W}(y|\bar{x})^s } \nonumber \\
	&= \log \frac{|\Xc|}{\kappa_{\!s}} + \log \widehat{W}(y|x)^s.%
\end{align}
Therefore,%
\begin{align}
	I^{\text{ML}}_{s,0}(\QX, \widehat{W}) &= \EX_{\QX \times \widehat{W}}[i_{s, 0}(X, Y)] \nonumber \\
	&= \log \frac{|\Xc|}{\kappa_{\!s}} + s \cdot \EX_{\widehat{W}_{\text{sym}}}[\log \widehat{W}_{\text{sym}}]
\end{align}
and the expansion terms proportional to $\sqrt{r}$ are%
\begin{align}\hspace{-0.5em}
	V_{s,0}(\QX, \widehat{W}) &= \EX_{\QX} [\VAR_{\widehat{W}} [i_{s,0}(X,Y)|X]] \nonumber \\
	 &= \EX_{\QX} [\VAR_{\widehat{W}} [s\cdot \log \widehat{W}(Y|X)|X]] \nonumber \\
	&=s^2 \! \left( 
	\EX_{\widehat{W}_{\text{sym}}}^{} \hspace{-0.6em}[\log^2 \widehat{W}_{\text{sym}}] {-} \EX_{\widehat{W}_{\text{sym}}}^2 \hspace{-0.6em}[ \log \widehat{W}_{\text{sym}} ] \right)\! . \label{App:eqn:generic_additive_channel_v_function_proof}
\end{align}

\subsection{Computations relative to the Gallager $E_0$ function}
The mismatched exponent density reads%
\begin{align}
	\varepsilon_{s, \zerov, \zerov, \rho}(x, y) &= \left( \frac{\frac{1}{|\Xc|} {\sum_{\bar{x}}}  \widehat{W}(y|\bar{x})^s}{\widehat{W}(y|x)^s} \right)^\rho \notag \\
	&= \frac{\kappa_{\!s}^\rho}{|\Xc|^\rho} \widehat{W}(y|x)^{-s\rho}.%
\end{align}
Therefore%
\begin{align}
	E_{s, \zerov, \zerov, \rho}^{\text{\normalfont ML}}(\QX, \widehat{W}) &= -\log \EX_{\QX{\times}\widehat{W}}[\varepsilon_{s, \zerov, \zerov, \rho}(X,Y)] \nonumber\\
	&= \log  \frac{|\Xc|^\rho}{\kappa_{\!s}^\rho \kappa_{1{-}s\rho}}
\end{align}
and the expansion terms proportional to $\sqrt{r}$ are%
\begin{align}
	V_{s,\zerov, \zerov, \rho}(\QX, \widehat{W}) 
	&= \frac{\EX_{\QX} [ \VAR_{\widehat{W}}[ \varepsilon_{s, \zerov, \zerov, \rho}(X, Y)|X]]}{\EX_{\QX{\times}\widehat{W}}^2[\varepsilon_{s, \zerov, \zerov, \rho}(X,Y)]} \nonumber \\
	&= \frac{\EX_{\QX} [ \VAR_{\widehat{W}}[ \widehat{W}(Y|X)^{- s \rho}|X]]}{\EX_{\QX{\times}\widehat{W}}^2[\widehat{W}(Y|X)^{- s \rho}]} \nonumber \\
	&= \frac{\kappa_{1-2s\rho} - \kappa_{1-s\rho}^2}{\kappa_{1{-}s\rho}^2} .
\end{align}

\section{Proof of Theorem \ref{theorem:CCcontinuous}} \label{App:ProofCont}
The continuous-alphabet case is similar to the discrete-alphabet case, but now we work with probability density functions instead of probability mass functions, and integrals instead of summations. Therefore, we resort to the calculus of variations to address the minimization in the function space.

We want to solve the following variational problem%
\begin{align}
	\min_{W \in \mathcal{B}} {\iint_{\mathbb{R}^2}} \QX(x) W(y|x) i_{s,\lv}(x, y) \DD x \DD y%
\end{align}
where the constraint $W \in \mathcal{B}$ considers conditional probability density functions $W(y|x) \geq 0$ that satisfy%
\begin{align}
	\!\left\{\!\!\! \begin{array}{rl}
		d(\widehat{W}\|W|\QX) + o(d(\widehat{W}\|W|\QX)\!) &\leq\ r\\
	{\int_{\mathbb{R}}} W(y|x) \DD y &=\ 1\ \text{for}\ x \in \mathbb{R}
	\end{array}\right.
\end{align}
with%
\begin{align}
	d(\widehat{W}\|W|\QX) \triangleq\! {\iint_{\mathbb{R}^2}} \QX(x) \frac{(W(y|x) {-} \widehat{W}(y|x)\!)^2}{2\widehat{W}(y|x)} \DD x \DD y.%
\end{align}
For sufficiently small $r$, the same trick as in \eqref{Appth1:Lprobl2} can be applied to move the $o(\cdot)$ term in the constraint as a penalty term in the objective function. We will omit the technical derivations, as it follows the same rules outlined in Appendix \ref{ProofLMRate}, and focus on solving the following variational problem%
\begin{align}
	\min_{\substack{d(\widehat{W}\|W|\QX) = r\\ {\int_{\mathbb{R}}} W(y|x)\DD y = 1}} {\iint_{\mathbb{R}^2}} \!\QX(x) W(y|x) i_{s,\lv}(x, y) \DD x \DD y + \rho^{\!\ast} \!o(r)%
\end{align}
where the differentiability of the objective functional and constraints is assumed under an appropriately chosen norm for channels with continuous derivatives, as functionals are either linear or quadratic in $W$.

Now, we proceed by minimizing the Lagrangian%
\begin{align}
	\mathcal{L}[W] = \iint_{\mathbb{R}^2} \!F(x, y, W(y|x)) \DD x \DD y%
\end{align}
with%
\begin{align} \label{App:eqn:F}
	F(x, y, W) \triangleq\ & \QX(x) Wi_{s,\lv}(x, y) + \nu(x) W \nonumber\\
	&+ \rho \QX(x)\frac{(W {-} \widehat{W}(y|x)\!)^2}{2\widehat{W}(y|x)}%
\end{align}
and where we have obviated $W$-independent constant terms. 
We follow the rules in \cite[Sec. V]{gelfand2000calculus}, where the stationary point $W^{\ast}(y|x)$ is found by solving the Euler-Lagrange equation%
\begin{align}
	\frac{\partial F}{\partial W} - \frac{\partial}{\partial x} \frac{\partial F}{\partial W_{\!x}} - \frac{\partial}{\partial y} \frac{\partial F}{\partial W_{\!y}} =0 \quad \text{in} \quad (x,y)\in\mathbb{R}^2.%
\end{align}
In our case, since $F$ in \eqref{App:eqn:F} does not depend on the partial derivatives of $W$, denoted $W_{x}$ and $W_{y}$, the above reduces to solving $\tfrac{\partial F}{\partial W} = 0$ in $(x,y)\in\mathbb{R}^2$ for $\{W^{\ast}(y|x), \nu^{\!\ast}(x), \rho^{\!\ast}\}$ jointly with the problem constraints.
Specifically, we have%
\begin{align}
	W^{\ast}(y|x) &= \widehat{W}(y|x) \!\left( 1- \frac{\QX(x)i_{s,\lv}(x,y){+}\nu^{\!\ast}(x)}{\rho^{\!\ast} \QX(x)} \right)
\end{align}
with%
\begin{align}
	\nu^{\!\ast}(x) &\triangleq -\QX(x)\EX_{\widehat{W}}[i_{s,\lv}(x, Y)]\\
	\rho^{\!\ast} &\triangleq \sqrt{V_{s, \lv}(\QX, \widehat{W})/(2r)}\\
	V_{s,\lv} (\QX, \widehat{W}) &\triangleq \EX_{\QX} [\VAR_{\widehat{W}}[i_{s,\lv}(X,Y)|X]].%
\end{align}
The stationary point $W^{\ast}(y|x)$ corresponds to the global minimum as the problem is linear in $W$ subject to quadratic constraints in $W$ and $\rho^{\ast}>0$.

\section{Gaussian Codebooks and Nearest Neighbor Decoding} \label{app:aux}
For Gaussian i.i.d. codebooks $\QX(x) = \mathcal{N}(x; P)$ and the nearest neighbor decoder $\widehat{W}(y|x) = \mathcal{N}(y{-}x;\hat{\sigma}^2)$, we have
\begin{align} \label{eqn:AppGaussian:Whats}
	\widehat{W} (y|x)^s &= \alpha_s \cdot \mathcal{N}(y{-}x; \hat{\sigma}^2 s^{-1})
	\\\label{eqn:AppGaussian:EWhats}
	\EX_{\QX}[\widehat{W}(y|X)^s] &= \alpha_s \cdot \mathcal{N}(y;\hat{\sigma}^2 s^{-1}{+}P)
\end{align}
with $\alpha_s \triangleq \sqrt{s^{-1} (2\pi\hat{\sigma}^2)^{1-s}}$.%

\subsection{Achievable rate with Gaussian i.i.d. codebooks}
The mismatched information density is (\textit{cf}. \cite[(28)-(29)]{7605463})%
\begin{align}
	i_{s,\zerov}(x, y) &= C_1 - \frac{s}{2}\frac{(y-x)^2}{\hat{\sigma}^2} + \frac{1}{2} \frac{y^2}{\hat{\sigma}^2 s^{-1}{+}P}
\end{align}
with
\begin{align}
	C_1 &\triangleq \frac{1}{2} \log \left(1+\frac{sP}{\hat{\sigma}^2}\right),%
\end{align}
and the respective expectations are%
\begin{align}
	\EX_{\widehat{W}} [i_{s,\zerov}(x, Y)] &= C_1 + \frac{1}{2} \frac{x^2{-}sP}{\hat{\sigma}^2 s^{-1}{+}P} \\
	\EX_{\QX {\times} \widehat{W}} [i_{s,\zerov}(X, Y)] &= C_1 + \frac{P}{2} \frac{1-s}{\hat{\sigma}^2 s^{-1}{+}P}.%
\end{align}
Note that $I^\text{\normalfont ML}_{s, \zerov}(\QX, \widehat{W}) =\EX_{\QX{\times}\widehat{W}} [i_{s,\zerov}(X, Y)]$.

Now we turn our attention to the term $V_{s,\zerov}(\QX,\widehat{W}) = \EX_{\QX {\times} \widehat{W}} [i_{s,\zerov}^2(X, Y)] - \EX_{\QX}[\EX_{\widehat{W}}^2[i_{s,\zerov}(X, Y)|X]]$. After tedious but otherwise straightforward calculations, we find its two terms as follows 
\begin{align}
	\EX_{\QX {\times} \widehat{W}} [i_{s,\zerov}^2(X, Y)] =\ & K_1 + \frac{3 s^2}{4} + \frac{3(\hat{\sigma}^2{+}P)^2}{4(\hat{\sigma}^2 s^{-1}{+}P)^2} \nonumber\\
	&- \frac{1}{2} \frac{s(3\hat{\sigma}^2{+}P)}{(\hat{\sigma}^2 s^{-1}{+}P)^2}\\
	\EX_{\QX}[\EX_{\widehat{W}\!}^2[i_{s,\zerov}(X, Y)|X]] =\ & K_1 + \frac{P^2(3-2s+s^2)}{4(\hat{\sigma}^2 s^{-1}{+}P)^2}
\end{align}
where $K_{1} \triangleq C_1^2 + C_1P(1{-}s)(\hat{\sigma}^2 s^{-1}{+}P)^{-2}$, leading to%
\begin{align}
	V_{s,\zerov}(\QX,\widehat{W}) &=\frac{P}{2} \frac{2\hat{\sigma}^2{+}s^2P}{(\hat{\sigma}^2 s^{-1}{+}P)^2}.
\end{align}

\subsection{$E_0$ function with Gaussian i.i.d. codebooks}
\label{app:exp_gauss}
The expressions for the worst-case mismatched Gallager $E_0$ functions can be obtained from%
\begin{align}
	\EX_{\QX}^t[\widehat{W}(y|X)^{s}] &\triangleq \alpha_s^{t} \frac{\sqrt{2\pi P_{\!s} t^{-1}}}{\sqrt{(2\pi P_{\!s})^t}} \cdot \mathcal{N}(y; P_{\!s} t^{-1})%
\end{align}
with $P_{\!s} \triangleq \hat{\sigma}^2 s^{-1}{+}P$, and%
\begin{align}
	A(s, t) &\triangleq \! \int_{\mathbb{R}} \EX_{\QX}^t [\widehat{W}(y|X)^{s}] \EX_{\QX}[\widehat{W}(y|X)^{1-st}] \DD y \nonumber\\
	&= \!\int_{\mathbb{R}}\! \frac{\alpha^t_s \alpha_{1-st}}{\sqrt{(2\pi P_{\!s})^{t}}} \sqrt{\!\frac{P_{\!s}}{P_{\!s} {+} t P_{\!1-st}}} \mathcal{N}\! \left(\! y; \frac{P_s P_{1-st}}{P_{s} {+} t P_{1-st}} \right) \! \DD y \nonumber\\
	&= \frac{\alpha^t_s \alpha_{1-st}}{\sqrt{(2\pi P_{\!s})^t}} \sqrt{\!\frac{P_{\!s}}{P_{\!s} + t P_{\!1-st}}} \nonumber\\
	%&= \sqrt{\frac{1}{(P_{\!s}s)^t (1-st)} {\cdot} \frac{P_{\!s}}{P_{\!s} + t P_{\!1-st}}} \nonumber\\
	%&= \sqrt{\frac{1}{(1+Ps)^t} {\cdot} \frac{1}{1+ \frac{tP(1-s-ts)}{P+s^{-1}}}}\ .
	&= \sqrt{\left(1+\frac{Ps}{\hat{\sigma}^2}\right)^{\!\!-t} \! \left( 1+ \frac{tP(1-s-ts)}{P+\hat{\sigma}^2 s^{-1}} \right)^{\!\!-1}}.
\end{align}
Therefore%
\begin{align}
	E^{\text{ML}}_{s,\zerov, \zerov, \rho}(\QX, \widehat{W}) =\ & -\log A(s, \rho) \nonumber \\
	=\ & \frac{\rho}{2} \log\!\left(1+\frac{Ps}{\hat{\sigma}^2}\right) \nonumber\\
	&+ \frac{1}{2} \log \left( 1+ \frac{\rho P (1{-}s{-}\rho s)}{P+\hat{\sigma}^2 s^{-1}} \right).
\end{align}

The dispersion term $V_{s, \zerov, \zerov, \rho}(\QX, \widehat{W})$ is computed as the following sum:
\begin{align}
	V_{s, \zerov, \zerov, \rho}(\QX, \widehat{W}) = V_1(\QX, \widehat{W})-V_2(\QX, \widehat{W}).
\end{align}
The first term is%
\begin{align}
	V_{1}(\QX, \widehat{W}) &\triangleq \frac{\EX_{\QX {\times} \widehat{W}} [ \varepsilon^2_{s, \zerov, \zerov,\rho}(X, Y)]}{\EX_{\QX{\times}\widehat{W}}^2[\varepsilon_{s, \zerov, \zerov, \rho}(X,Y)]} \nonumber\\
	&= \frac{A(s, 2\rho)}{A(s, \rho)^2} \nonumber \\%
	&= \frac{1}{\sqrt{P_s}} \frac{P{+}\hat{\sigma}^2 s^{-1} + \rho P(1{-}s{-}\rho s)}{\!\sqrt{P{+}\hat{\sigma}^2s^{-1} + 2\rho P (1{-} s {-} 2\rho s)}}.
\end{align}
The second term is%
\begin{align}
	V_2(\QX, \widehat{W}) &= \frac{\EX_{\QX}[\EX_{\widehat{W}}^2 [ \varepsilon_{s, \zerov, \zerov,\rho}(X, Y)|X]]}{\EX_{\QX{\times}\widehat{W}}^2[\varepsilon_{s, \zerov, \zerov, \rho}(X,Y)]}.
\end{align}
We compute the numerator through the steps%
\begin{align}
	&\EX_{\widehat{W}}^2 [\varepsilon_{s, \zerov, \zerov, \rho}(x,Y)] =  \frac{\alpha_{1-\rho s}^2 \alpha^{2\rho}_s}{\rho(2\pi P_{\!s})^{\rho-1}}  \mathcal{N}^2\!\left(\! x; \frac{P_s {-} \rho s P}{\rho(1{-}s\rho)} \right)\\
	&\EX_{\QX}\EX_{\widehat{W}}^2 [\varepsilon_{s, \zerov, \zerov, \rho}(X,Y)|X] = \frac{1}{s \rho} \frac{\sqrt{2}(\hat{\sigma}^{2} s^{-1})^{\rho} P_s^{1-\rho}}{P_s {+} \rho P (2{-}s{-}2s\rho)}
\end{align}
and thus%
\begin{align}
	V_2(\QX, \widehat{W}) &= \frac{\sqrt{2}}{\rho s} \frac{P{+}\hat{\sigma}^2s^{-1} + \rho P(1{-}s{-}\rho s)}{P{+}\hat{\sigma}^2 s^{-1} + \rho P (2{-}s{-}2\rho s)}.
\end{align}

\subsection{Gaussian spherical codebooks}
For spherical codebooks, we have%
\begin{align}
	\EX_{\QX}[\widehat{W}^s(y|X)e^{\lambda\! X^2}] = \frac{\alpha}{\sqrt{g(P)}} \cdot \mathcal{N}(y;\hat{\sigma}^2 s^{-1} {+} P')%
\end{align}
and the auxiliary variables%
\begin{align}
	C_2 &\triangleq \frac{1}{2} \log \left(g(P)+\frac{sP}{\hat{\sigma}^2}\right)\\%
	P' &\triangleq P/g(P)\\
	g(P) &\triangleq 1-2\lambda P.
\end{align}

The mismatched information density is%
\begin{align}
	i_{s,\lambda}(x, y) &= \lambda x^2+ C_2 - \frac{s}{2}\frac{(y{-}x)^2}{\hat{\sigma}^2} + \frac{1}{2} \frac{y^2}{\hat{\sigma}^2 s^{-1}{+}P'}%
\end{align}
and%
\begin{align}
	\EX_{\widehat{W}} [i_{s,\lambda}(x, Y)] &= \lambda x^2 + C_2 + \frac{1}{2} \frac{x^2{-}sP'}{\hat{\sigma}^2 s^{-1}{+}P'} \\
	\EX_{\QX {\times} \widehat{W}} [i_{s,\lambda}(X, Y)] &= \lambda P + C_2 + \frac{1}{2} \frac{P-sP'}{\hat{\sigma}^2 s^{-1}{+}P'}.%
\end{align}
After straightforward calculations, we can write%
\begin{align}
	I^\text{\normalfont ML}_{s, \lambda}(\QX, \widehat{W}) &= \frac{1}{2} \log\!\bigg(g(P)+\frac{sP}{\hat{\sigma}^2}\bigg) \nonumber \\
	&+ \frac{P(1{-}s){+}\hat{\sigma}^2 s^{-1}g(P)(1{-}g(P))}{2(P+\hat{\sigma}^2 s^{-1}g(P))}\\
	V_{s,\lambda} (\QX, \widehat{W}) &= \frac{P}{2} \frac{2\hat{\sigma}^2 g^2(P){+}s^2P}{(\hat{\sigma}^2 s^{-1}g(P) + P)^2}.
\end{align}

%\cleardoublepage
\bibliographystyle{IEEEtran}
\bibliography{finalreport}

\begin{thebibliography}{10}
\providecommand{\url}[1]{#1}
\csname url@samestyle\endcsname
\providecommand{\newblock}{\relax}
\providecommand{\bibinfo}[2]{#2}
\providecommand{\BIBentrySTDinterwordspacing}{\spaceskip=0pt\relax}
\providecommand{\BIBentryALTinterwordstretchfactor}{4}
\providecommand{\BIBentryALTinterwordspacing}{\spaceskip=\fontdimen2\font plus
\BIBentryALTinterwordstretchfactor\fontdimen3\font minus
  \fontdimen4\font\relax}
\providecommand{\BIBforeignlanguage}[2]{{%
\expandafter\ifx\csname l@#1\endcsname\relax
\typeout{** WARNING: IEEEtran.bst: No hyphenation pattern has been}%
\typeout{** loaded for the language `#1'. Using the pattern for}%
\typeout{** the default language instead.}%
\else
\language=\csname l@#1\endcsname
\fi
#2}}
\providecommand{\BIBdecl}{\relax}
\BIBdecl

\bibitem{GallagerWiley1968}
R.~G. Gallager, \emph{{I}nformation {T}heory and {R}eliable
  {C}ommunication}.\hskip 1em plus 0.5em minus 0.4em\relax John Wiley \& Sons
  Inc., 1968.

\bibitem{lapidoth1998reliable}
A.~Lapidoth and P.~Narayan, ``Reliable communication under channel
  uncertainty,'' \emph{IEEE Transactions on Information Theory}, vol.~44,
  no.~6, pp. 2148--2177, 1998.

\bibitem{370120}
I.~Csisz{\'a}r and P.~Narayan, ``Channel capacity for a given decoding
  metric,'' \emph{IEEE Trans. Inf. Theory}, vol.~41, no.~1, pp. 35--43, 1995.

\bibitem{Foundations}
J.~Scarlett, A.~Guill{\'e}n~i F{\`a}bregas, A.~Somekh-Baruch, and A.~Martinez,
  ``{I}nformation-{T}heoretic {F}oundations of {M}ismatched {D}ecoding,''
  \emph{{Foundations and Trends\textsuperscript{\textregistered} in
  Communications and Information Theory}}, vol.~17, no. 2--3, pp. 149--401,
  2020.

\bibitem{EuclideanInformationTheory}
S.~Borade and L.~Zheng, ``Euclidean information theory,'' in \emph{2008 IEEE
  International Zurich Seminar on Communications}, 2008, pp. 14--17.

\bibitem{ConvexOptimization}
S.~Boyd and L.~Vandenberghe, \emph{{C}onvex {O}ptimization}.\hskip 1em plus
  0.5em minus 0.4em\relax Cambridge University Press, 2004.

\bibitem{gelfand2000calculus}
I.~M. Gelfand, R.~A. Silverman \emph{et~al.}, \emph{Calculus of
  variations}.\hskip 1em plus 0.5em minus 0.4em\relax Courier Corporation,
  2000.

\bibitem{MismatchedHypTesting}
P.~Boroumand and A.~Guill{\'e}n~i F{\`a}bregas, ``{M}ismatched {H}ypothesis
  {T}esting: {E}rror {E}xponent {S}ensitivity,'' \emph{IEEE Transactions on
  Information Theory}, vol.~68, pp. 6738--6761, 10 2022.

\bibitem{kaplan1993ira}
G.~Kaplan and S.~Shamai, ``{Information rates and error exponents of compound
  channels with application to antipodal signaling in a fading environment},''
  \emph{AEU. Archiv f{\"u}r Elektronik und {\"U}bertragungstechnik}, vol.~47,
  no.~4, pp. 228--239, 1993.

\bibitem{Hui83}
J.~Hui, ``Fundamental issues of multiple accessing,'' \emph{PhD dissertation,
  MIT}, 1983.

\bibitem{CsiszarKorner81graph}
I.~Csisz{\'{a}}r and J.~K{\"o}rner, ``Graph decomposition: A new key to coding
  theorems,'' \emph{IEEE Transactions on Information Theory}, vol.~27, no.~1,
  pp. 5--12, 1981.

\bibitem{UpperboundEhsan}
E.~Asadi~Kangarshahi and A.~Guill{\'e}n~i F{\`a}bregas, ``A single-letter upper
  bound to the mismatch capacity,'' \emph{{IEEE} Transactions on Information
  Theory}, vol.~67, no.~4, pp. 2013--2033, 2021.

\bibitem{MulticastAnelia}
A.~Somekh-Baruch, ``A single-letter upper bound on the mismatch capacity via
  multicast transmission,'' \emph{{IEEE} Transactions on Information Theory},
  vol.~68, no.~5, pp. 2801--2812, 2022.

\bibitem{6763080}
J.~Scarlett, A.~Martinez, and A.~Guill{\'e}n~i F{\`a}bregas, ``Mismatched
  decoding: Error exponents, second-order rates and saddlepoint
  approximations,'' \emph{IEEE Transactions on Information Theory}, vol.~60,
  no.~5, pp. 2647--2666, 2014.

\bibitem{ElementsOfInformationTheory}
T.~M. Cover and J.~A. Thomas, \emph{{E}lements of {I}nformation
  {T}heory}.\hskip 1em plus 0.5em minus 0.4em\relax John Wiley \& Sons Inc.,
  2006.

\bibitem{669134}
A.~Lapidoth and E.~Telatar, ``The compound channel capacity of a class of
  finite-state channels,'' \emph{IEEE Transactions on Information Theory},
  vol.~44, no.~3, pp. 973--983, 1998.

\bibitem{7464362}
S.~Loyka and C.~D. Charalambous, ``A general formula for compound channel
  capacity,'' \emph{IEEE Transactions on Information Theory}, vol.~62, no.~7,
  pp. 3971--3991, 2016.

\bibitem{532892}
A.~Lapidoth, ``Nearest neighbor decoding for additive non-{Gaussian} noise
  channels,'' \emph{IEEE Transactions on Information Theory}, vol.~42, no.~5,
  pp. 1520--1529, 1996.

\bibitem{BoundsMeasures}
A.~L. Gibbs and F.~E. Su, ``On choosing and bounding probability metrics,''
  \emph{International Statistical Review / Revue Internationale de
  Statistique}, vol.~70, no.~3, pp. 419--435, 2002.

\bibitem{FDivergence}
S.~M. Ali and S.~D. Silvey, ``A general class of coefficients of divergence of
  one distribution from another,'' \emph{Journal of the Royal Statistical
  Society: Series B (Methodological)}, vol.~28, no.~1, pp. 131--142, 1966.

\bibitem{FDivergence2}
I.~Csisz{\'a}r, ``Information-type measures of difference of probability
  distributions and indirect observation,'' \emph{Studia Scientarum
  Mathematicarum Hungarica}, vol.~2, pp. 229--318, 1967.

\bibitem{RenyiEntropy}
P.~Harremo\"{e}s and T.~van Erven, ``R{\'e}nyi divergence and
  {K}ullback-{L}eibler divergence,'' \emph{IEEE Transactions on Information
  Theory}, vol.~60, pp. 3797--3820, 7 2014.

\bibitem{1054753}
S.~Arimoto, ``An algorithm for computing the capacity of arbitrary discrete
  memoryless channels,'' \emph{IEEE Transactions on Information Theory},
  vol.~18, no.~1, pp. 14--20, 1972.

\bibitem{Minimax}
K.~Fan, ``{M}inimax {T}heorems,'' \emph{Proceedings of the National Academy of
  Sciences of the United States of America}, vol.~39, no.~1, pp. 42--47, 1953.

\bibitem{10619112}
F.~Molina, P.~Patel, and A.~{Guill{\'e}n i F{\`a}bregas}, ``Nearest neighbor
  decoding for a class of compound channels,'' in \emph{2024 IEEE International
  Symposium on Information Theory (ISIT)}, 2024, pp. 1510--1513.

\bibitem{340469}
N.~Merhav, G.~Kaplan, A.~Lapidoth, and S.~Shamai~Shitz, ``On information rates
  for mismatched decoders,'' \emph{IEEE Transactions on Information Theory},
  vol.~40, no.~6, pp. 1953--1967, 1994.

\bibitem{7605463}
J.~Scarlett, V.~Tan, and G.~Durisi, ``The dispersion of nearest-neighbor
  decoding for additive non-{Gaussian} channels,'' \emph{IEEE Transactions on
  Information Theory}, vol.~63, no.~1, pp. 81--92, 2017.

\end{thebibliography}

\end{document}